\documentclass[preprint,12pt]{elsarticle}

\usepackage[T1]{fontenc}
\usepackage[utf8]{inputenc}
\usepackage{amsmath,amssymb,amsthm}
\usepackage{booktabs}
\usepackage{array}
\newcolumntype{L}[1]{>{\raggedright\arraybackslash}p{#1}}
\usepackage{multirow}
\usepackage{graphicx}
\usepackage{xcolor}
\usepackage{tikz}
\usetikzlibrary{positioning,arrows.meta,shapes.geometric,fit,backgrounds,calc}
\definecolor{crestblue}{HTML}{2A78D6}
\definecolor{crestorange}{HTML}{EB6834}
\usepackage{hyperref}

\newtheorem{proposition}{Proposition}
\journal{Vehicular Communications}

\begin{document}

\begin{frontmatter}

\title{From Trust Scores to Driver Decisions: An Abstaining Evidential Decision Layer for Event Validation in VANETs}

\author{Ipek Abas{\i}kele\c{s}-Turgut\corref{cor1}}
\ead{ipek.abasikeles@iste.edu.tr}
\cortext[cor1]{Corresponding author.}
\affiliation{organization={Computer Engineering Department, Iskenderun Technical University},
            city={Iskenderun, Hatay},
            postcode={31200},
            country={T\"urkiye}}

\begin{abstract}
Trust management in vehicular ad hoc networks is highly active, yet the decision that a driver needs, whether a reported event is real, has received far less attention; the schemes that make it force every decision into two outcomes, even on thin or conflicting evidence. This paper presents CREST, a vehicle-side event-decision framework that can withhold a decision its evidence does not support. Each vehicle weights the reports it holds by the trust, proximity and freshness of their senders, fuses them with Yager's rule, and decides only when the fused margin reaches an abstention threshold. Abstention is a waiting state: a vehicle still unsure at the event is resolved by an explicit application policy, which makes the trade-off between missed hazards and unnecessary detours visible. CREST requires only a scalar trust per reporter and is built here on an evidential trust model whose consensus and revocation steps are updated to support the decision layer. A single report is shown to change the decision margin by at most its weight, which bounds the influence of an attacker coalition by its trust. In a pre-registered simulation study with fourteen attacks, including position spoofing and trust-gated camouflage, and with all unspecified parameters of all systems tuned by the same procedure, CREST reduces the decision error of two published trust-based dissemination systems from 0.162 and 0.106 to 0.051. With the best-guess policy it also lowers both missed and unnecessary reroutes and outperforms the forced binary rule on both error types in every cell.
\end{abstract}


\end{frontmatter}

\section{Introduction}\label{sec:intro}

Safety applications in vehicular networks depend on what vehicles tell each other. A vehicle that reports an accident, a flooded road or a blocked lane allows the vehicles behind it to slow down or to change their route before they can see the hazard themselves. The same channel can be abused by insiders: a vehicle with valid credentials may hide a real hazard or announce one that does not exist, and the consequences range from congestion to accidents~\cite{fabi2022forest,tian2024aka}. Authentication schemes ensure that a message comes from a legitimate vehicle and has not been modified~\cite{zhang2022tapd,ghani2025vec,sawant2026qssnet}, but they cannot tell whether the reported event is real.

Trust management addresses this gap and has become one of the most active topics in vehicular security. Recent proposals evolve trust formulas with genetic programming~\cite{aslan2023dynamic}, organise trust evaluation in social communities~\cite{fabi2022forest}, store reputation in blockchains~\cite{shen2024cache,peng2024blockchain,han2024decentralized}, combine trust with incentive mechanisms~\cite{sharma2026pcv,han2024decentralized}, and fuse trust evidence with Dempster--Shafer theory~\cite{bhargava2022duel,ipek}. Most of these models end with a trust score: a value per vehicle or per message, a set of revoked vehicles, or a selected service provider~\cite{sharma2026pcv,gupta2025oteet,shen2024cache}. A driver, however, does not act on a trust score but on a decision about the event: is the reported hazard there or not. The step from the trust values of several, possibly conflicting reports to this decision receives far less attention than the trust computation itself. It is often assumed to be solved, as when trust evaluation starts ``once a true emergency situation is validated''~\cite{fabi2022forest}, and Sharma and Thakur describe an earlier trust model as one that ``focuses on trust observation and not on decision making''~\cite{sharma2026pcv}.

The schemes that do define an event decision share three limitations (Section~\ref{sec:related}). First, they force every decision into two outcomes, even when the evidence is thin or contradictory~\cite{tcemd,htemd,atwa2021rteam,arnda,zhang2026mdf}; one of them explicitly counts an undecided case as unacceptable~\cite{atwa2021rteam}, and another rejects waiting for confirmations as too slow~\cite{zhang2026mdf}. A forced decision converts the lack of evidence into false alarms or missed hazards and lets a single report tip a balanced case. Second, the reports are weighted by the proximity of their senders to the event, but the positions are self-declared, and position falsification is not evaluated~\cite{arnda,han2024decentralized,aslan2023dynamic}. Third, the evaluation stops at trust values or decision accuracy and does not show what a decision costs the driver, i.e.\ how many hazards are missed and how many detours are unnecessary.

This paper proposes CREST, a vehicle-side event-decision framework that allows a vehicle to withhold a decision that its evidence does not support. Each vehicle weights the reports it holds by the trust of their senders, by their proximity to the event and by their freshness, and fuses them with Yager's rule, which moves conflicting evidence to ignorance instead of renormalising it away. The vehicle decides only when the fused margin between the two hypotheses reaches an abstention threshold; otherwise it waits for more evidence. Abstention is thereby a waiting state, not an outcome: a vehicle that is still unsure when it reaches the event resolves the case by an explicit application policy, and the cost of remaining unsure is measured rather than hidden. CREST requires only a scalar trust value per reporter and is instantiated here on the evidential trust model IPEK~\cite{ipek}, whose consensus and revocation steps are updated to support the decision layer.

The contributions of this paper are as follows.
\begin{itemize}
  \item \textbf{An abstaining evidential decision framework} for event validation on the vehicle, with explicit policies for vehicles that remain unsure, so that the trade-off between missed hazards and unnecessary detours can be chosen by the application (Section~\ref{sec:crest}).
  \item \textbf{Analytical properties} that link the trust layer to the robustness of the decision: a single report changes the decision margin by at most its weight, so that reversing a correct decision requires a coalition whose trust-bounded weights reach twice the abstention threshold; the conditions under which an honest vehicle abstains; and the breakdown point of the credibility-weighted consensus (Section~\ref{sec:analysis}).
  \item \textbf{Two updates of the trust substrate} motivated by the decision layer: a credibility-weighted median consensus that keeps the attackers' trust low under local attacker majorities, and a time-based revocation rule whose delay does not depend on the update tempo (Section~\ref{sec:crest:substrate}).
  \item \textbf{An evaluation of complete systems} against HTEMD~\cite{htemd} and TCEMD~\cite{tcemd}, and of decision rules including AR-NDA~\cite{arnda} on a common trust substrate, under fourteen attacks, among them position spoofing and trust-gated camouflage, at two traffic densities and four decision geometries. The protocol was registered before the runs, all unspecified parameters of all systems were tuned with the same procedure, and the results are reported at the trust, decision and application levels (Sections~\ref{sec:method} and~\ref{sec:results}).
\end{itemize}

The remainder of the paper is organised as follows. Section~\ref{sec:related} reviews trust management and event decision schemes. Section~\ref{sec:system} defines the system, the decision geometry and the threat model. Section~\ref{sec:crest} presents CREST, and Section~\ref{sec:analysis} derives its analytical properties. Section~\ref{sec:method} describes the evaluation methodology, and Section~\ref{sec:results} the results. Section~\ref{sec:discussion} discusses limitations, and Section~\ref{sec:conclusion} concludes the paper.

\section{Related work}\label{sec:related}

\subsection{Trust management in vehicular networks}\label{sec:related:trust}

Trust in vehicular networks is used in at least three different senses. In authentication and key-agreement schemes, trust is the property of a credential or of the infrastructure: a receiver accepts a message when its signature, pseudonym and freshness are valid~\cite{zhang2022tapd,tian2024aka,ghani2025vec,sawant2026qssnet}. Such schemes establish who sent a message, but not whether its content is true, and a vehicle with valid credentials can still report a false event. Several of them take trust values as an external input~\cite{zhang2022tapd} or leave misbehaviour detection to future work~\cite{sawant2026qssnet}. In resource-oriented schemes, a reputation score selects service providers, such as parked computing vehicles~\cite{sharma2026pcv}, task-offloading servers~\cite{gupta2025oteet} or cache providers~\cite{shen2024cache}, by comparing the score with a threshold. In such schemes trust drives a selection rather than an assessment of content.

The third sense, which is the one addressed here, is trust in the content of event reports. Entity-centric models rate the senders, data-centric models rate the messages, and hybrid models combine both~\cite{tcemd,htemd,aslan2023dynamic}. Recent work has made the trust computation itself increasingly elaborate: genetic programming evolves the trust formula from many local evidences and adapts it when traffic conditions change~\cite{aslan2023dynamic}; community detection and a forest-fire spreading model select the vehicles that aggregate trust and relay emergency messages~\cite{fabi2022forest}; blockchains store reputation and game-theoretic incentives discourage false reports~\cite{peng2024blockchain,han2024decentralized}; and evidential models fuse feedback with Dempster--Shafer theory~\cite{bhargava2022duel,ipek}. The output of these models is a trust score or a set of revoked vehicles. How a receiving vehicle should turn the scores of several, possibly conflicting, reports into a decision about the event is often taken for granted. Fabi and Thampi, for instance, start the trust evaluation ``once a true emergency situation is validated''~\cite{fabi2022forest}, without specifying the validation.

IPEK~\cite{ipek}, the trust substrate used in this work, belongs to the evidential family. It scales rewards and penalties with event severity and location criticality, discounts feedback from reporters that deviate systematically from the consensus of witnesses, and fuses the feedback with Yager's rule into a global trust. Like most trust models, IPEK produces trust values and revocation decisions at the CA and does not decide on the vehicle whether an event is real.

\subsection{Event decision schemes}\label{sec:related:decision}

A smaller group of schemes defines an explicit decision about the event. Table~\ref{tab:related} compares them with CREST.

\paragraph{Trust-weighted voting}
TCEMD~\cite{tcemd} weights the $\pm 1$ opinion of each report by the CA-certified trust of its sender. A vehicle stores the messages it receives until it reaches a fixed decision distance and then follows the sign of the weighted vote, even on a single message; low-confidence results are acted upon but not rebroadcast. HTEMD~\cite{htemd} maps each message to an event probability from the hybrid certificate and direct trust of the sender, shrinks the mean probability towards 0.5 in proportion to the entropy of the message set, and compares the result with a threshold. Han et al.~\cite{han2024decentralized} accept a road-event message when a weighted sum of sender reputation, an exponential decay with the distance to the event and the similarity to already accepted messages exceeds a threshold. In all three schemes the decision is a threshold on a trust-weighted score. Freshness is at most a validity cut-off, and uncertainty is either not represented or, in HTEMD, represented as a neutral score that is still thresholded. A consistent colluding majority lowers the entropy in HTEMD instead of raising it.

\paragraph{Risk- and learning-based rules}
RTEAM~\cite{atwa2021rteam} estimates, from trust-, hop-, time- and location-weighted report mass, the likelihood that the event is real and takes the action with the lowest Bayes risk, with error costs that depend on the time to reach the event and on the number of following vehicles. RTEAM is one of the few schemes with asymmetric decision costs, but it regards the case ``where a vehicle could not determine the event state'' as ``unacceptable'' and is designed to decide in all cases. MDF-TDM~\cite{zhang2026mdf} decides on each event message from the reputation of the sender, two feedback-driven risk scores and a context-predicted prior probability of the event, and explicitly argues against waiting for confirmations, since ``waiting \ldots\ may result in significant delays''. Aslan and Sen~\cite{aslan2023dynamic} evolve the data-trust formula with genetic programming over 18 local evidences, including proximity, recency and the share of vehicles reporting the same event, and accept or reject each message by a threshold. These rules use richer evidence than voting, but each message or event is still forced into one of two classes, and the cases with thin evidence, which Aslan and Sen identify as the hardest~\cite{aslan2023dynamic}, are decided like all others.

\paragraph{Evidential rules}
Dempster--Shafer theory can represent the absence of evidence explicitly as mass on the whole frame. Peng et al.~\cite{peng2024blockchain} fuse reputation-weighted votes with Dempster's rule in a hierarchical blockchain, but assign mass only to singleton hypotheses, so that no ignorance is represented, and select the level with the largest belief. DUEL~\cite{bhargava2022duel} introduces an explicit uncertainty level into the trust opinion and divides its decision ratio by the uncertainty mass, so that a receiver may act at once, verify with neighbours or discard a message; the uncertainty, however, concerns the sender rather than the event, and under high uncertainty the receiver ``is bound to react according to the majority''. AR-NDA~\cite{arnda} is the most complete evidential scheme. The RSU weights each opinion by a fuzzy trustworthiness of reputation and spatio-temporal proximity, fuses the opinions with Dempster's rule, and switches to a vote-count and role-aware inference when the fused masses differ by less than a tolerance. The uncertainty thus changes the rule but never the binary form of the output; the decision is centralised at the RSU, and the authors state that they do not consider location forgery, although proximity is an input of the trust.

\subsection{Abstention and robustness}\label{sec:related:abstention}

The option to withhold a decision is classical in pattern recognition. Chow~\cite{chow1970} showed that a classifier that rejects inputs whose posterior falls below a threshold trades a lower error rate on the accepted inputs against the rate of rejected ones. Two properties distinguish the vehicular setting. First, the evidence of a vehicle approaching an event grows with time, so that a withheld decision can be taken later on more evidence; abstention is a waiting state rather than a final outcome. Second, the vehicle must eventually act, so that the cost of remaining unsure has to be counted as a missed hazard or as a detour. The schemes of Section~\ref{sec:related:decision} wait at most for a fixed distance or time window and then force a decision~\cite{tcemd,htemd,aslan2023dynamic}; none of them defers a decision because its evidence is insufficient, and none reports what deferral costs at the application level.

Two further aspects concern robustness. The decision schemes that weight reports by proximity~\cite{atwa2021rteam,han2024decentralized,arnda,aslan2023dynamic} rely on positions declared by the reporters; position falsification is a well-known misbehaviour in cooperative intelligent transportation systems~\cite{vanderheijden2019survey}, but none of the schemes evaluates it, and AR-NDA excludes it explicitly~\cite{arnda}. In addition, consensus-based filters such as the median of IPEK~\cite{ipek} or the outlier filter of HTEMD~\cite{htemd} break down once colluding attackers form a local majority among the reporters of a target~\cite{huber2009}.

\begin{table}[t]
\centering
\caption{Comparison of event decision schemes. Site: where the event decision is taken (V: vehicle, R: RSU, B: blockchain miners). Weights of the reports: T trust, P proximity, F freshness. Unc.: uncertainty represented and used in the decision. Abst.: the decision can be deferred because of insufficient evidence. Rev.: revocation of attackers. Attacks evaluated: Hid. hiding true events, Fab. fabricated events, Col. collusion, Pos. position falsification. Metrics: Dec. decision-level, App. application-level. \checkmark: yes; $\sim$: partly; --: no.}
\label{tab:related}
\scriptsize
\setlength{\tabcolsep}{3pt}
\begin{tabular}{@{}lcccccccccccccc@{}}
\toprule
& & & \multicolumn{3}{c}{Weights} & & & & \multicolumn{4}{c}{Attacks} & \multicolumn{2}{c}{Metrics} \\
\cmidrule(lr){4-6}\cmidrule(lr){10-13}\cmidrule(l){14-15}
Scheme & Year & Site & T & P & F & Unc. & Abst. & Rev. & Hid. & Fab. & Col. & Pos. & Dec. & App. \\
\midrule
TCEMD~\cite{tcemd}                     & 2020 & V & \checkmark & --         & $\sim$     & --         & -- & \checkmark & \checkmark & \checkmark & \checkmark & -- & \checkmark & -- \\
RTEAM~\cite{atwa2021rteam}             & 2021 & V & \checkmark & \checkmark & \checkmark & --         & -- & --         & \checkmark & --         & --         & -- & \checkmark & -- \\
DUEL~\cite{bhargava2022duel}           & 2022 & V & \checkmark & --         & $\sim$     & $\sim$     & -- & --         & \checkmark & \checkmark & --         & -- & --         & -- \\
Aslan \& Sen~\cite{aslan2023dynamic}   & 2023 & V & \checkmark & \checkmark & \checkmark & --         & -- & --         & --         & \checkmark & --         & -- & \checkmark & -- \\
HTEMD~\cite{htemd}                     & 2024 & V & \checkmark & --         & --         & $\sim$     & -- & \checkmark & --         & \checkmark & \checkmark & -- & \checkmark & -- \\
Peng et al.~\cite{peng2024blockchain}  & 2024 & B & \checkmark & --         & --         & --         & -- & --         & --         & \checkmark & --         & -- & \checkmark & -- \\
Han et al.~\cite{han2024decentralized} & 2024 & V & \checkmark & \checkmark & $\sim$     & --         & -- & \checkmark & --         & \checkmark & --         & -- & $\sim$     & -- \\
AR-NDA~\cite{arnda}                    & 2026 & R & \checkmark & \checkmark & \checkmark & \checkmark & -- & --         & \checkmark & \checkmark & \checkmark & -- & \checkmark & -- \\
MDF-TDM~\cite{zhang2026mdf}            & 2026 & V & \checkmark & --         & --         & --         & -- & --         & --         & \checkmark & \checkmark & -- & \checkmark & -- \\
\midrule
CREST (this work)                      & --   & V & \checkmark & \checkmark & \checkmark & \checkmark & \checkmark & \checkmark & \checkmark & \checkmark & \checkmark & \checkmark & \checkmark & \checkmark \\
\bottomrule
\end{tabular}
\end{table}

\subsection{Positioning}\label{sec:related:position}

Table~\ref{tab:related} shows a consistent pattern. Every scheme that decides about an event forces the decision into two outcomes. Explicit uncertainty, where present, either concerns the sender~\cite{bhargava2022duel} or selects a fallback rule~\cite{arnda}. Proximity weights rest on unverified positions, and the evaluation stops at decision accuracy or trust values without measuring what the decision costs the driver. CREST addresses these gaps on the vehicle side: it represents the lack of evidence explicitly, defers decisions that the evidence does not support, makes the cost of remaining unsure visible through explicit application policies, bounds the influence of an attacker coalition analytically, and is evaluated against complete competing systems at the trust, decision and application levels, including position spoofing and trust-gated camouflage.

\section{System and threat model}\label{sec:system}

CREST is a vehicle-side decision framework that determines whether a reported event is genuine from the reports received by the vehicle and the trust of their senders. The framework does not define a trust model of its own; in common with decision schemes that rely on an external reputation source~\cite{arnda}, it requires only a scalar trust value in $[0,1]$ for each reporter. In the present work, this trust value is provided by IPEK~\cite{ipek}, whose system, event and perception models are adopted without modification. This section summarises the elements of the trust substrate on which the decision layer depends and defines the decision geometry and the threat model. The values of all parameters are listed in Tables~\ref{tab:params} and~\ref{tab:params-baselines} (Section~\ref{sec:method:setup}).

\subsection{Network, events and perception}\label{sec:system:net}

The system consists of vehicles, road-side units (RSUs) and a central authority (CA). The CA computes the global trust of every vehicle and the revocation list; RSUs relay feedback to the CA and the trust state back to the vehicles. The CA and the RSUs are trusted. Vehicles have unique identities and know their own positions. Message delivery within the broadcast range $r_{\mathrm{b}}$ is assumed reliable, so that the evaluation isolates the trust and decision logic from channel effects.

$N_E$ traffic events are placed at fixed, randomly drawn locations, equally many in each quadrant of the map. Each event $e$ has an event severity $S_E$ and a location criticality $S_L$, combined into the criticality factor $CF = S_E + S_L - S_E S_L$. Events belong to three classes: non-critical (Class~1, both dimensions low), mixed (Class~2, one dimension high) and critical (Class~3, both dimensions high). An event is active for $D_e = D_0(1 + CF)$, then passive, and it reactivates at the same location after a class-dependent random delay. The location criticality is static map information and is therefore known to every vehicle; the event severity and the current state are known only to vehicles that observe the event.

Every event has an impact radius $R$. A vehicle within $R$ is a \emph{witness}: it perceives the state of the event with error probability
\begin{equation}\label{eq:perr}
  p_{\mathrm{err}}(d) = p_0 \, \frac{d}{R},
\end{equation}
where $d$ is its distance from the event and $p_0$ the base error probability, and it broadcasts an event report every $T_{\mathrm{rep}}$ while it remains within $R$. A report carries the sender identity, the event identifier, the reported state (active or not active), the sender position and a timestamp. A receiver keeps only the most recent report of each sender per event, discards reports of revoked senders, reports received farther than the acceptance distance $D_{\mathrm{acc}}$ from the event, and reports older than $D_e/2$.

\subsection{Trust substrate}\label{sec:system:trust}

IPEK~\cite{ipek} computes trust in two steps. A witness $V_j$ compares each recorded report of a reporter $V_i$ with its own perception and updates a local trust $LT_{ji} \in [0,1]$. Agreement is rewarded with rate $r^{+} = \mu\,\omega$, where $\omega = \omega_E S_E + \omega_L S_L$; disagreement is penalised with rate $r^{-} = \min(1, \kappa\,r^{+}\,CF)\,\gamma_{ji}$, where the penalty ratio $\kappa > 1$ makes trust harder to gain than to lose and the observation-reliability factor $\gamma_{ji}$ discounts disagreements that may be caused by perception errors. The local trust values are sent to the CA as feedback.

At the end of each round of length $\Delta$, the CA first updates the credibility $c_j$ of every reporter from its mean absolute deviation from a per-target consensus of the feedback, where the deviation $d_0$ maps to zero credibility and $\rho$ is the smoothing factor. It then converts each feedback item into a simple mass function over \{trusted, risky\}, discounted by $w_j = GT_j\,c_j$, fuses the items of each target sequentially with Yager's rule~\cite{yager1987}, combines the result with the previous state, applies risk accentuation, and obtains the scalar global trust $GT_i$ by the pignistic transformation. A vehicle whose global trust stays below a detection threshold $\theta$ is revoked; IPEK uses $N_r$ consecutive rounds for this test.

The decision layer of CREST adopts the global trust $GT_i$ as the trust of reporter $V_i$ (Section~\ref{sec:crest:decision}). Two steps of the substrate, namely the consensus used in the credibility test and the revocation rule, are updated in this work (Section~\ref{sec:crest:substrate}); all remaining components and parameters of IPEK are retained as published~\cite{ipek}. The trust-layer parameters were re-tuned under the fairness protocol of Section~\ref{sec:method}, which applies the same procedure to every compared system. The tuning retained the published penalty ratio $\kappa$ and changed the two credibility parameters $d_0$ and $\rho$, which IPEK had fixed without tuning (Table~\ref{tab:params}).

\subsection{Decision geometry}\label{sec:system:geometry}

A trust substrate such as IPEK serves the CA by determining which vehicles are trustworthy and which are to be revoked. A driver, in contrast, requires a decision about the event before reaching it, and this decision is made on the vehicle side. A vehicle approaching an event becomes a \emph{decider} when it enters the \emph{decision ring} $R < d \le \beta R$, where $\beta > 1$ is the ring ratio; its task is to decide whether the event is active before it enters the impact radius, where it would observe the event itself (Fig.~\ref{fig:geometry}).

\begin{figure}[t]
\centering
\begin{tikzpicture}[
    font=\footnotesize,
    veh/.style={line width=0.5pt, minimum width=7mm, minimum height=3.2mm, rounded corners=1pt, inner sep=0pt},
    wit/.style={veh, draw=black!60, fill=black!20},
    decd/.style={veh, draw=crestblue, fill=crestblue!30},
    lis/.style={veh, draw=black!35, fill=black!8},
    lab/.style={font=\scriptsize, text=black!80},
    rep/.style={-{Stealth[length=1.6mm]}, line width=0.5pt, draw=crestblue!80},
  ]
  \fill[black!6] (-7.2,-0.32) rectangle (4.6,0.32);
  \draw[black!35, dash pattern=on 3pt off 3pt] (-7.2,0) -- (4.6,0);
  \draw[black!35, dashed, line width=0.4pt] (-6.58,-2.39) arc[start angle=200, end angle=160, radius=7.0];
  \node[lab, align=center] at (-6.15,2.85) {acceptance\\distance $D_{\mathrm{acc}}$};
  \fill[crestblue!10, even odd rule] (0,0) circle (3.3) (0,0) circle (2.0);
  \draw[crestblue!70, line width=0.5pt] (0,0) circle (3.3);
  \draw[black!55, line width=0.6pt] (0,0) circle (2.0);
  \node[star, star points=7, star point ratio=0.55, draw=crestorange!90!black, fill=crestorange!70, minimum size=5mm, inner sep=0pt] (ev) at (0,0) {};
  \node[lab, below=1mm of ev] {event};
  \draw[{Stealth[length=1.4mm]}-{Stealth[length=1.4mm]}, line width=0.4pt, black!70] (0,-0.05) -- node[lab, left, pos=0.55] {$R$} (-1.25,-1.56);
  \draw[{Stealth[length=1.4mm]}-{Stealth[length=1.4mm]}, line width=0.4pt, crestblue!80] (0.05,0.05) -- node[lab, right, pos=0.78] {$\beta R$} (2.03,2.60);
  \node[lab, align=center] at (0.95,-1.5) {impact radius\\(witnesses)};
  \node[lab, align=center, text=crestblue!90!black] at (0,-2.65) {decision ring\\$R < d \le \beta R$};
  \node[wit] (w1) at (0.95,0.16) {};
  \node[wit] (w2) at (-0.7,-0.16) {};
  \node[wit] (w3) at (0.3,1.15) {};
  \node[decd] (d) at (-2.65,0.16) {};
  \node[lab, above=0.6mm of d, text=crestblue!90!black] {decider};
  \node[lis] (l) at (-5.0,-0.16) {};
  \node[lab, below=0.8mm of l, align=center] {receives and\\stores reports};
  \draw[rep] (w2.west) to[bend right=12] (d.east);
  \draw[rep] (w3.west) to[bend right=18] (d.north east);
  \draw[rep] (w1.south west) to[bend left=22] (d.south east);
  \draw[rep, draw=black!30] (w2.south west) to[bend left=14] (l.east);
  \node[lab, text=crestblue!90!black] at (-1.25,0.95) {reports};
  \draw[-{Stealth[length=1.8mm]}, line width=0.6pt, black!60] (-4.05,0.55) -- (-3.4,0.55);
\end{tikzpicture}
\caption{Decision geometry (schematic, not to scale). Witnesses inside the impact radius $R$ observe the event and broadcast reports; a vehicle in the decision ring $R < d \le \beta R$ decides from the reports it holds before it reaches the event. Vehicles farther away store reports up to the acceptance distance $D_{\mathrm{acc}}$. The schematic was drawn with code prepared with the assistance of an AI tool (see the declaration of generative AI use).}
\label{fig:geometry}
\end{figure}
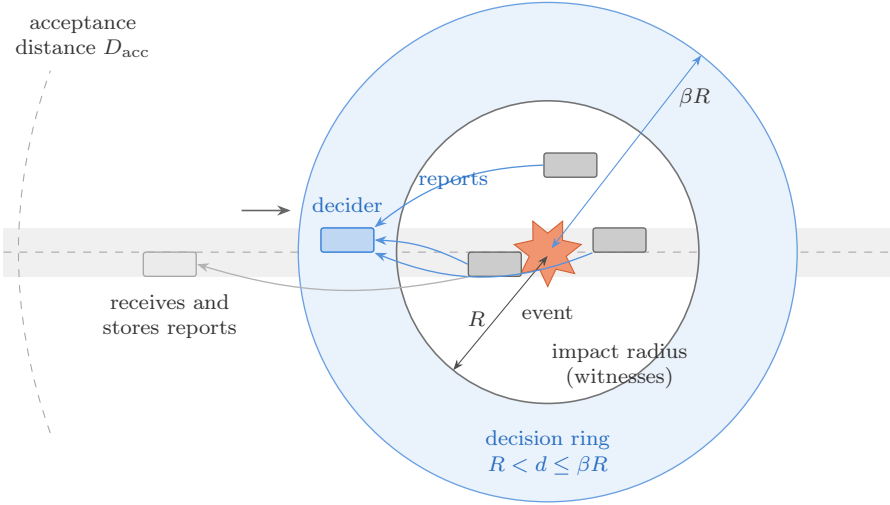
The decider bases its decision on the reports it holds for the event (Section~\ref{sec:system:net}). Optionally, a vehicle that has reached a firm decision forwards the original reports it used, unchanged and with their original timestamps and positions, for up to $h$ hops; forwarded reports are used for decisions only and never as trust evidence. The main evaluation uses no relaying ($h = 0$); a larger ring ratio $\beta_{\mathrm{L}}$ and relaying over $h_{\max}$ hops are evaluated as a sensitivity analysis.

A decider that leaves the decision ring without having received any report about the event is \emph{uninformed}. Uninformed vehicles occur in every compared system, are a property of the geometry and the traffic density, and do not change their route in any system.

\subsection{Threat model}\label{sec:system:threat}

The adversary is an insider with a valid identity that follows the protocol in form and falsifies its content. Each vehicle becomes an attacker independently with probability $a \in \mathcal{A}$, and all attackers know each other and collude. An attacker acts on up to three channels: the \emph{report} channel (the state it broadcasts), the \emph{feedback} channel (the local trust values it sends to the CA, set to the extreme values 0 for honest targets and 1 for accomplices) and, newly in this work, the \emph{position} channel (the position it declares in its reports). Outside its trigger condition the attacker behaves honestly and accumulates reputation.

Table~\ref{tab:attacks} lists the fourteen attacks. The eight attacks of IPEK combine four triggers with two channel configurations~\cite{ipek}. The severity trigger attacks events with $S_E \ge \theta_E$, the location trigger events with $S_L \ge \theta_L$, the on-off trigger attacks in the first half of every period $T_{\mathrm{oo}}$, and the continuous trigger attacks at every opportunity; in these attacks an attacker reports an active event as not active. CREST adds three families that target the decision layer:
\begin{itemize}
  \item \textbf{Fabrication.} An attacker that witnesses a passive event reports it as active, always (Fab) or within the on-off window (Fab-OnOff). This attack targets false alarms, the error type that hiding attacks do not produce.
  \item \textbf{Position spoofing.} The attacker hides active severe events as in R-Sev, but declares a position drawn uniformly from a disk of radius $r_{\mathrm{s}}$ around the event centre, so that its report appears to come from the most reliable location. In addition, a spoofing attacker that is not a witness but is inside the decision ring sends the same false report \emph{remotely}. The attack exploits any decision rule that rewards proximity.
  \item \textbf{Trust-gated camouflage.} The attacker uses the severity trigger but attacks only while its own global trust satisfies $GT \ge \theta + \delta$ for a margin $\delta > 0$; otherwise it behaves honestly to recover. This attacker knows the detection threshold $\theta$ and stays above it by design, which makes it the strongest adversary for any threshold-based revocation.
\end{itemize}
Sybil identities, external attackers, silent attackers and attacks on the CA or the RSUs are outside the scope of this work.

\begin{table}[t]
\centering
\caption{Attacks used in the evaluation. Channels: R = report, F = feedback, P = position.}
\label{tab:attacks}
\small
\begin{tabular}{@{}L{3.2cm}L{1.6cm}L{3.3cm}L{4.2cm}@{}}
\toprule
Attack & Channels & Trigger & Falsification \\
\midrule
D-Sev, D-Loc, D-OnOff, D-Cont & R + F & severity / location / on-off / continuous & active reported as not active \\
R-Sev, R-Loc, R-OnOff, R-Cont & R & severity / location / on-off / continuous & active reported as not active \\
Fab, Fab-OnOff & R & passive event / on-off & passive reported as active \\
Spoof, Spoof-OnOff & R + P & severity / on-off & as R-Sev, position near the centre, remote reports \\
Camo-R, Camo-D & R / R + F & severity and $GT \ge \theta + \delta$ & as R-Sev / D-Sev \\
\bottomrule
\end{tabular}
\end{table}

\section{CREST}\label{sec:crest}

CREST is an event-decision framework executed on every vehicle. Its core is an abstaining evidential decision rule (Section~\ref{sec:crest:decision}), under which a vehicle fuses the reports it holds, weighted by the trust of their senders, commits to a firm decision only when the accumulated evidence is sufficiently strong, and otherwise defers the decision until further evidence arrives. Since the rule requires only a scalar trust value for each reporter, it is not tied to a particular trust model; in this study, the trust values are supplied by IPEK.

The robustness of the decision nevertheless depends on the quality of the trust substrate: according to Proposition~\ref{prop:reverse}, the coalition required to reverse a decision grows as the trust of the attackers decreases. For this reason, two steps of the substrate are additionally updated (Section~\ref{sec:crest:substrate}). The first is the consensus of the credibility test, which breaks down under local attacker majorities and thereby allows the trust of the attackers to rise; the second is the revocation rule, whose delay depends on the update tempo. These updates apply to IPEK rather than to the decision framework and are therefore evaluated separately through ablation. An optional plausibility check, which is not part of the CREST core, is described in Section~\ref{sec:crest:plaus}. Fig.~\ref{fig:architecture} gives an overview of the framework and of its interface to the trust substrate.

\begin{figure}[t]
\centering
\resizebox{\linewidth}{!}{\begin{tikzpicture}[
    font=\scriptsize,
    dbox/.style={draw=crestblue, fill=crestblue!8, line width=0.5pt, rounded corners=2pt, align=center, minimum height=9mm, text width=21mm, inner sep=2pt},
    sbox/.style={draw=black!55, fill=black!4, line width=0.5pt, rounded corners=2pt, align=center, minimum height=9mm, text width=23mm, inner sep=2pt},
    ubox/.style={sbox, dashed, draw=crestorange!90!black, fill=crestorange!8},
    dec/.style={draw=crestblue, fill=white, line width=0.5pt, diamond, aspect=2.0, align=center, inner sep=0.5pt},
    arr/.style={-{Stealth[length=1.5mm]}, line width=0.5pt, draw=black!70},
    lab/.style={font=\scriptsize, text=black!70},
    title/.style={font=\footnotesize\bfseries, anchor=north west},
  ]
  \node[dbox] (rep) at (0,0)   {Reports held\\$(s_i, t_i, \mathbf{x}_i)$};
  \node[dbox] (w)   at (2.7,0) {Weights\\$w_i = GT_i\,\pi_i\,\phi_i$};
  \node[dbox] (fus) at (5.4,0) {Yager fusion\\$g = M^{\mathrm{A}} - M^{\mathrm{N}}$};
  \node[dec]  (q)   at (7.85,0) {$|g| \ge \tau$};
  \node[dbox] (firm) at (10.4,0.85)  {Firm decision:\\reroute or keep route};
  \node[dbox] (def)  at (10.4,-0.85) {Defer};
  \node[dbox, text width=31mm] (pol) at (10.4,-2.45) {Unsure on leaving\\the ring: no action,\\CREST+T or CREST+I};
  \draw[arr] (rep) -- (w);
  \draw[arr] (w) -- (fus);
  \draw[arr] (fus) -- (q);
  \draw[arr] (q.north) |- node[pos=0.3, left, lab] {yes} (firm.west);
  \draw[arr] (q.south) |- node[pos=0.3, left, lab] {no} (def.west);
  \draw[arr, dashed] ([yshift=-2.5mm]def.west) -| node[pos=0.3, below, lab] {re-evaluate every $T_{\mathrm{u}}$} (fus.south);
  \draw[arr] (def) -- (pol);
  \node[draw=crestblue!60, line width=0.5pt, rounded corners=3pt, fit=(rep)(firm)(pol), inner xsep=3mm, inner ysep=3mm, yshift=2.5mm] (veh) {};
  \node[title, text=crestblue!90!black] at (veh.north west) {Decision framework (every vehicle)};

  \node[sbox] (fb)   at (0,-4.9)    {Feedback $LT_{ji}$\\from witnesses};
  \node[ubox] (cred) at (2.85,-4.9) {Credibility $c_j$:\\weighted-median\\consensus};
  \node[sbox] (gt)   at (5.7,-4.9)  {Global trust $GT_i$:\\Yager, history,\\risk accentuation};
  \node[ubox] (rev)  at (8.55,-4.9) {Revocation:\\time rule ($W$, $\theta$)};
  \draw[arr] (fb) -- (cred);
  \draw[arr] (cred) -- (gt);
  \draw[arr] (gt) -- (rev);
  \node[draw=black!45, line width=0.5pt, rounded corners=3pt, fit=(fb)(rev), inner xsep=3mm, inner ysep=3mm, yshift=2.5mm] (ca) {};
  \node[title, text=black!75] at (ca.north west) {Trust substrate at the CA (IPEK)};

  \draw[arr, line width=0.7pt, draw=crestblue!90!black] (ca.north -| gt) -- node[right, lab, text=crestblue!90!black] {$GT_i$, revocation list} (veh.south -| gt);
  \draw[arr, draw=black!50] (veh.south -| rep) -- node[right, lab, align=left] {feedback of\\witnesses} (ca.north -| rep);

  \node[ubox, minimum height=3mm, text width=5mm, anchor=north west] (lg) at ([yshift=-2.5mm]ca.south west) {};
  \node[anchor=west, lab] at (lg.east) {update of the trust substrate in this work};
\end{tikzpicture}}
\caption{Overview of CREST. On every vehicle, the reports held about an event are weighted by the trust of their senders, by proximity and by freshness and fused with Yager's rule; a firm decision is taken only when the margin reaches the abstention threshold $\tau$, otherwise the decision is deferred. The trust substrate at the CA supplies the global trust values and the revocation list; the two dashed components are the updates of the substrate introduced in this work. The schematic was drawn with code prepared with the assistance of an AI tool (see the declaration of generative AI use).}
\label{fig:architecture}
\end{figure}
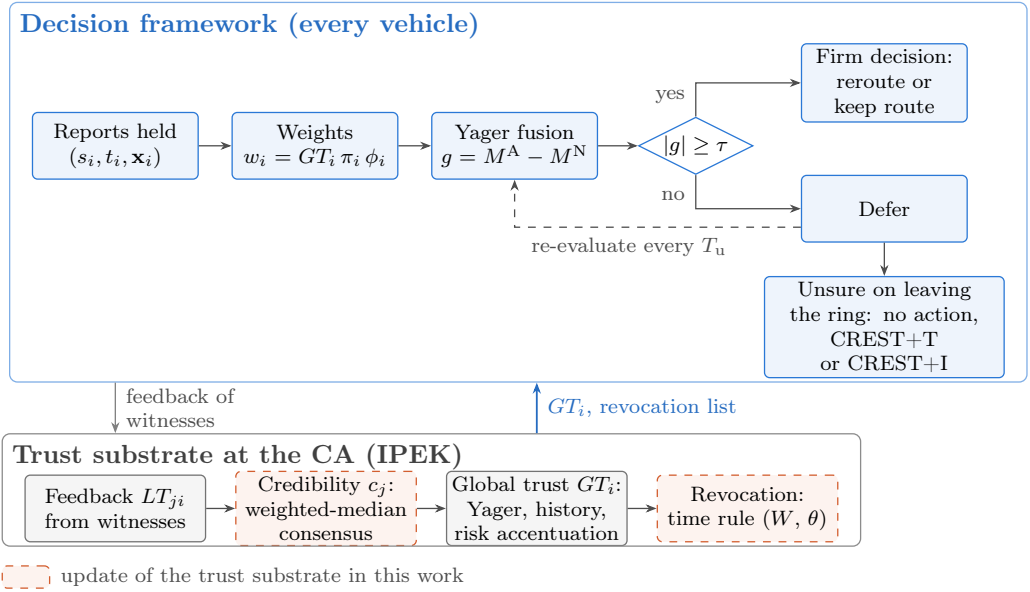

\subsection{Abstaining evidential decision layer}\label{sec:crest:decision}

\paragraph{Report weights}
Let a decider hold the reports $\{(s_i, t_i, \mathbf{x}_i)\}_{i=1}^{n}$ about event $e$, where $s_i \in \{\mathrm{A}, \mathrm{N}\}$ is the reported state (active or not active), $t_i$ the timestamp and $\mathbf{x}_i$ the declared reporter position. Each report receives the weight
\begin{equation}\label{eq:weight}
\begin{gathered}
  w_i = GT_i \cdot \pi_i \cdot \phi_i, \\
  \pi_i = \max\!\left(0,\, 1 - \frac{\lVert \mathbf{x}_i - \mathbf{x}_e \rVert}{R}\right), \qquad
  \phi_i = \max\!\left(0,\, 1 - \frac{t - t_i}{D_e/2}\right),
\end{gathered}
\end{equation}
where $GT_i$ is the current global trust of the reporter as broadcast by the CA, $\mathbf{x}_e$ the event location and $t$ the current time. The proximity factor $\pi_i$ follows the perception model~\eqref{eq:perr}: a report from the edge of the impact radius is the least reliable. The freshness factor $\phi_i$ decreases to zero at the report-age limit $D_e/2$ (Section~\ref{sec:system:net}), so that a report about to expire contributes little. All three factors lie in $[0,1]$, hence $w_i \in [0,1]$.

\paragraph{Evidential fusion}
Each report is a simple support function on the frame $\{\mathrm{A}, \mathrm{N}\}$: it assigns mass $w_i$ to the reported state and $1 - w_i$ to the whole frame $\Theta$, which represents ignorance. The reports are fused sequentially with Yager's rule~\cite{yager1987}, in the order in which their senders were first heard (a newer report of the same sender replaces the older one in place). For two mass functions $m_1, m_2$ with components $(m^{\mathrm{A}}, m^{\mathrm{N}}, m^{\Theta})$,
\begin{equation}\label{eq:yager}
\begin{aligned}
  m^{\mathrm{A}} &= m_1^{\mathrm{A}} m_2^{\mathrm{A}} + m_1^{\mathrm{A}} m_2^{\Theta} + m_1^{\Theta} m_2^{\mathrm{A}}, \\
  m^{\mathrm{N}} &= m_1^{\mathrm{N}} m_2^{\mathrm{N}} + m_1^{\mathrm{N}} m_2^{\Theta} + m_1^{\Theta} m_2^{\mathrm{N}}, \\
  m^{\Theta} &= m_1^{\Theta} m_2^{\Theta} + K, \qquad K = m_1^{\mathrm{A}} m_2^{\mathrm{N}} + m_1^{\mathrm{N}} m_2^{\mathrm{A}}.
\end{aligned}
\end{equation}
Unlike Dempster's rule~\cite{shafer1976}, Yager's rule does not renormalise the conflict $K$ away but moves it to ignorance. Contradicting reports therefore do not yield a confident result but increase the uncertainty, which is precisely the situation in which the decider ought to defer. The same rule is employed by IPEK for global trust~\cite{ipek}; in CREST it is applied to event reports rather than to trust feedback.

\paragraph{Decision rule}
Let $M = (M^{\mathrm{A}}, M^{\mathrm{N}}, M^{\Theta})$ be the fused state and $g = M^{\mathrm{A}} - M^{\mathrm{N}} \in [-1, 1]$ the \emph{decision margin}. The decider takes a firm decision when the margin is large enough:
\begin{equation}\label{eq:decision}
  \text{decision} =
  \begin{cases}
    \text{active} & \text{if } g \ge \tau, \\
    \text{not active} & \text{if } g \le -\tau, \\
    \text{defer} & \text{otherwise},
  \end{cases}
\end{equation}
with the abstention threshold $\tau \in (0, 1)$. A firm decision is final for the current activation of the event. A deferred decision is re-evaluated at every position update of the decider, i.e.\ every $T_{\mathrm{u}}$, with the reports held at that moment, until it either becomes firm or the decider leaves the decision ring, by entering the impact radius or by driving away. A decider that leaves the ring without a firm decision but with at least one report is \emph{unsure}; one without any report is uninformed (Section~\ref{sec:system:geometry}). Only firm decisions trigger relaying (Section~\ref{sec:system:geometry}).

The threshold $\tau$ bounds the influence of an attacker coalition: as shown in Section~\ref{sec:analysis}, a single report changes $g$ by at most its weight $w_i$, so that reversing a correct firm decision requires reports whose weights sum to at least $2\tau$. Since $w_i \le GT_i$, a coalition of attackers with low global trust must be large to reverse a decision. This links the quality of the trust layer to the robustness of the decision. A forced binary rule corresponds to $\tau = 0$ and loses this protection.

\paragraph{Abstention as a waiting state}
Abstention in CREST is a waiting state, not an outcome. A deferring vehicle keeps collecting reports, and its later firm decision rests on more evidence than an immediate forced one. If the vehicle is still unsure when it reaches the event, the application has to act, and CREST leaves this choice to an explicit policy. Three policies are considered:
\begin{itemize}
  \item \textbf{No action:} an unsure vehicle keeps its route. This minimises unnecessary reroutes but misses every active event on which the vehicle stays unsure.
  \item \textbf{CREST+T (best guess):} an unsure vehicle acts on the sign of its last margin $g$, i.e.\ it reroutes if $M^{\mathrm{A}} \ge M^{\mathrm{N}}$.
  \item \textbf{CREST+I (cautious):} an unsure vehicle reroutes, treating a potential hazard as present.
\end{itemize}
CREST+T is the default. CREST+I suits applications in which a missed hazard is much more costly than a detour. The three policies differ only for unsure vehicles; they share every firm decision, so they can be compared on the same simulation runs.

\subsection{Updates to the trust substrate}\label{sec:crest:substrate}

\subsubsection{Credibility-weighted median consensus}\label{sec:crest:cons}

The credibility test of IPEK compares the feedback of each reporter with a per-target consensus, computed as the plain median of the local trust values reported about the target in the round~\cite{ipek}. A plain median has a breakdown point of one half: once colluding attackers are a local majority among the reporters of a target, they determine the consensus. The credibility test then reverses: honest reporters appear deviant and lose credibility, while attackers appear consistent. This happens in practice at high attacker shares, because attackers are not spread evenly over the targets.

CREST replaces the plain median by the credibility-weighted median. For target $V_i$ with reported values $\{LT_{ji}\}_{j \in \mathcal{R}_i}$, the consensus is
\begin{equation}\label{eq:wmedian}
  \tilde{\ell}_i = \min\Bigl\{ \ell \in \{LT_{ji}\} \;:\; \sum_{j \in \mathcal{R}_i,\, LT_{ji} \le \ell} c_j \;\ge\; \tfrac{1}{2} \sum_{j \in \mathcal{R}_i} c_j \Bigr\},
\end{equation}
where $c_j$ is the credibility of reporter $V_j$ from the previous round; if the cumulative weight equals exactly one half, the two adjacent values are averaged, and if all weights are zero, the plain median is used. As in IPEK, the consensus is defined only for targets with at least $n_{\min}$ reporters, and revoked reporters are excluded. The deviations and the credibility update of IPEK are then computed against $\tilde{\ell}_i$.

The weighting moves the breakdown point from a share of reporters to a share of credibility. Attackers can control $\tilde{\ell}_i$ only when their total credibility reaches half of the total (Section~\ref{sec:analysis}). Because the credibility of attackers decreases as long as the consensus is honest, the margin of protection grows by itself. The weight uses credibility only: a variant weighted by $c_j\,GT_j$ performed worse in the tuning stage, since the global trust is itself corrupted at high attacker shares (Section~\ref{sec:discussion}).

\subsubsection{Time-based revocation}\label{sec:crest:revoke}

IPEK revokes a vehicle whose global trust stays below $\theta$ for $N_r$ consecutive rounds~\cite{ipek}. The resulting revocation delay is proportional to the round length $\Delta$, so any change of the update tempo changes the delay. CREST replaces the round count by a time window: a vehicle is revoked at the first round boundary at which its global trust has been below $\theta$ \emph{continuously} for at least $W$ seconds,
\begin{equation}\label{eq:revoke}
  \text{revoke } V_i \text{ at } t_k \iff GT_i(t_m) < \theta \text{ for all round boundaries } t_m \in [t_k - W, t_k].
\end{equation}
A round in which the global trust is at or above $\theta$ restarts the window. The delay then lies in $[W, W + \Delta)$ for any $\Delta$ (Section~\ref{sec:analysis}), so that the update tempo enters only as a rounding term. In this study the window equals $W = (N_r - 1)\Delta$ (Table~\ref{tab:params}), so that at the fixed round length the time rule yields the same decisions as the round rule; the benefit of the rule is that it decouples revocation from the tempo, not a different behaviour at the operating point.

\subsection{Optional position plausibility check}\label{sec:crest:plaus}

Position spoofing (Section~\ref{sec:system:threat}) exploits the proximity factor $\pi_i$ of \eqref{eq:weight}. As an optional add-on, a receiver can test every directly received report against two plausibility conditions:
\begin{itemize}
  \item \textbf{P1:} the declared reporter position is farther from the receiver than $r_{\mathrm{b}} + m_{\mathrm{p}}$, the broadcast range plus a margin;
  \item \textbf{P2:} $\lvert d_{\mathrm{decl}} - \hat{d} \rvert > \tau_{\mathrm{p}}\, \hat{d}$, where $d_{\mathrm{decl}}$ is the distance implied by the declared position and $\hat{d}$ a noisy distance estimate from the received signal, modelled as $\hat{d} = d\,(1 + \varepsilon)$ with $\varepsilon \sim \mathcal{N}(0, \sigma_{\mathrm{p}}^2)$, and $\tau_{\mathrm{p}}$ is the tolerance.
\end{itemize}
A report that fails either condition is dropped without penalising its sender. A penalising variant was rejected in the tuning stage because noise-induced rejections of honest reports reduced honest trust. Forwarded reports cannot be checked in this way, because the receiver has no signal from the original reporter. The check is therefore evaluated separately as an add-on (Section~\ref{sec:results}) and is not part of the CREST configuration compared in the main results.

\section{Analytical properties}\label{sec:analysis}

This section derives five properties that follow from the rules of the decision framework and of the substrate updates. Propositions~\ref{prop:influence} and~\ref{prop:reverse} bound the influence of a reporter on the decision and link the trust layer to the robustness of the decision layer. Proposition~\ref{prop:abstain} states when an honest decider abstains. Proposition~\ref{prop:revoke} shows that the revocation delay of CREST does not scale with the round length. Proposition~\ref{prop:wmedian} gives the breakdown point of the credibility-weighted consensus.

\subsection{Influence of a reporter on the decision}

Fusing a report with state $\mathrm{A}$ and weight $w$ into a state
$M = (M^{\mathrm{A}}, M^{\mathrm{N}}, M^{\Theta})$ with Yager's rule~\eqref{eq:yager} gives
\begin{equation}\label{eq:yager-step}
  (M^{\mathrm{A}}, M^{\mathrm{N}}, M^{\Theta}) \;\mapsto\; \bigl(M^{\mathrm{A}} + w M^{\Theta},\; (1-w) M^{\mathrm{N}},\; (1-w) M^{\Theta} + w M^{\mathrm{N}}\bigr),
\end{equation}
and symmetrically for a report with state $\mathrm{N}$. The map is linear in $M$, and starting from the vacuous state $(0, 0, 1)$ it also produces the first fused report. The fused state of a report list is therefore the product of linear maps applied to $(0,0,1)$.

\begin{proposition}[Single-reporter influence]\label{prop:influence}
Inserting one report of weight $w$ at any position of any report list changes the decision margin $g = M^{\mathrm{A}} - M^{\mathrm{N}}$ by at most $w$. A coalition of reports with weights $w_1, \dots, w_k$ changes it by at most $\min(2, \sum_j w_j)$.
\end{proposition}

\begin{proof}
Let the report be inserted after a prefix with fused state $s = (M^{\mathrm{A}}, M^{\mathrm{N}}, M^{\Theta})$, and assume it reports $\mathrm{A}$ (the other case is symmetric). By~\eqref{eq:yager-step} the insertion changes the state by
\[
  \delta = w\,\bigl(M^{\Theta},\; -M^{\mathrm{N}},\; M^{\mathrm{N}} - M^{\Theta}\bigr).
\] The suffix acts as a linear map $L$, so the final margin is $f \cdot s$ with the row vector $f = (f_1, f_2, f_3) = (1, -1, 0)\,L$, and the change of the margin is
\[
  f \cdot \delta = w\,\bigl[ M^{\Theta}(f_1 - f_3) + M^{\mathrm{N}}(f_3 - f_2) \bigr].
\]
It is shown by backward induction over the suffix that $f_1 \ge f_3 \ge f_2$, $f_1 - f_3 \le 1$ and $f_3 - f_2 \le 1$. For the empty suffix, $f = (1, -1, 0)$ satisfies these conditions. Prepending an $\mathrm{A}$ report of weight $u$ maps $f$ to $\bigl(f_1,\; (1-u) f_2 + u f_3,\; u f_1 + (1-u) f_3\bigr)$; hence $f_1 - f_3' = (1-u)(f_1 - f_3)$ and $f_3' - f_2' = u (f_1 - f_3) + (1-u)(f_3 - f_2)$, which preserve the order and both bounds. An $\mathrm{N}$ report is symmetric. Therefore $\lvert f \cdot \delta \rvert \le w\,(M^{\Theta} + M^{\mathrm{N}}) \le w$. The coalition bound follows by inserting the reports one at a time and applying the triangle inequality; since $g \in [-1, 1]$, no change exceeds 2.
\end{proof}

The bound does not depend on the position of the report, so it holds for every fusion order. A numerical check over $5 \cdot 10^4$ random report lists ($n \le 30$, random insertion positions) found no violation larger than $3 \cdot 10^{-16}$; a tighter coalition bound of the form $1 - \prod_j (1 - w_j)$ does not hold.

\begin{proposition}[Cost of reversing a decision]\label{prop:reverse}
Let a decider hold a correct firm decision with margin $g \ge \tau$. Pushing it into abstention requires coalition weight $\sum_j w_j > g - \tau$; reversing it into a wrong firm decision requires $\sum_j w_j \ge g + \tau \ge 2\tau$. Since $w_j \le GT_j$ by~\eqref{eq:weight}, $k$ attackers with global trust at most $T_a$ can reverse a correct firm decision only if $k \ge 2\tau / T_a$.
\end{proposition}

\begin{proof}
Immediate from Proposition~\ref{prop:influence}.
\end{proof}

Proposition~\ref{prop:reverse} is the quantitative link between the two layers: every reduction of the attackers' global trust by the trust layer enlarges the coalition that the decision layer can withstand. The two alternative rules lose this property in different ways. A forced binary rule has $\tau = 0$, so the barrier $2\tau$ disappears and a single report can tip a balanced list. A majority vote gives every report the same weight, so the influence of a report is independent of its trust.

\subsection{Abstention of honest deciders}

\begin{proposition}[Abstention condition]\label{prop:abstain}
If a decider holds $n$ concordant reports of equal weight $w \in (0, 1)$, its margin is $g = 1 - (1 - w)^n$, and it decides firmly if and only if
\[
  n \ge n^{*}(\tau, w) = \left\lceil \frac{\ln(1 - \tau)}{\ln(1 - w)} \right\rceil .
\]
\end{proposition}

\begin{proof}
By~\eqref{eq:yager-step} with $M^{\mathrm{N}} = 0$, the ignorance mass is multiplied by $1 - w$ at each step, so $M^{\Theta} = (1-w)^n$ and $g = M^{\mathrm{A}} = 1 - (1 - w)^n$. The condition $g \ge \tau$ is equivalent to the stated bound.
\end{proof}

A single report suffices if its weight satisfies $w \ge \tau$, whereas for a small threshold a weight of $\tau/2$ already requires three reports. Since $w$ decreases with the distance of the reporter from the event and with the age of the report, abstention among honest deciders is caused by sparse, distant or stale evidence. It therefore decreases with traffic density, and its application cost, the hazards missed by vehicles that remain unsure, is largest in sparse traffic. Section~\ref{sec:results} confirms both effects.

\subsection{Revocation delay}

\begin{proposition}[Tempo independence of time-based revocation]\label{prop:revoke}
Let the global trust of a vehicle fall below $\theta$ at round boundary $t_0$ and remain below it. Under the round rule with $N_r$ rounds the vehicle is revoked at $t_0 + (N_r - 1)\Delta$. Under the time rule~\eqref{eq:revoke} with window $W$ it is revoked at $t_0 + \lceil W / \Delta \rceil \Delta$, i.e.\ after a delay in $[W, W + \Delta)$.
\end{proposition}

\begin{proof}
Global trust and the revocation test change only at round boundaries $t_0 + k\Delta$. The round rule counts boundaries and fires at the $N_r$-th one, $k = N_r - 1$. The time rule fires at the first boundary with $k\Delta \ge W$, $k = \lceil W/\Delta \rceil$.
\end{proof}

Under the round rule the delay grows linearly with $\Delta$, so any change of the update tempo, for instance an adaptive tempo, also changes the revocation delay. Under the time rule $\Delta$ enters only as a rounding term smaller than one round. With the parameters of this study, $W = (N_r - 1)\Delta$ (Table~\ref{tab:params}), so that both rules revoke after exactly $W$, which is why the two rules produce identical results at the operating point of this study (Section~\ref{sec:results}). The proposition covers the revocation rule only: the convergence of global trust itself still depends on the number of rounds, so the detection delay as a whole is not tempo-free.

\subsection{Breakdown point of the weighted consensus}

\begin{proposition}[Breakdown point]\label{prop:wmedian}
Let the reporters of a target be split into honest reporters $\mathcal{H}$ and attackers $\mathcal{A}$, and let $\tilde{\ell}$ be the credibility-weighted median~\eqref{eq:wmedian}. If
\[
  \sum_{a \in \mathcal{A}} c_a < \frac{1}{2} \sum_{j \in \mathcal{H} \cup \mathcal{A}} c_j ,
\]
then $\tilde{\ell}$ lies within the range of the honest values, $\min_{h \in \mathcal{H}} LT_h \le \tilde{\ell} \le \max_{h \in \mathcal{H}} LT_h$, whatever values the attackers report. If all honest reporters have credibility $c_h$ and all attackers $c_a$, the consensus is safe as long as the local attacker share satisfies
\[
  p < p^{*} = \frac{c_h}{c_h + c_a}.
\]
\end{proposition}

\begin{proof}
Let $C$ be the total credibility. The cumulative credibility of all values up to the largest honest value contains the whole honest credibility, which exceeds $C/2$; hence the weighted median, the smallest value at which the cumulative credibility reaches $C/2$, is not larger than that value. When the cumulative credibility equals exactly $C/2$ at some value, the median is averaged with the next value of positive weight, which is again not larger than the largest honest value, since honest credibility remains above it. Conversely, the values below the smallest honest value carry only attacker credibility, which is less than $C/2$, so the median is not smaller than the smallest honest value. With uniform credibilities the condition reads $p\,c_a < (1-p)\,c_h$, which gives $p^{*}$.
\end{proof}

For the plain median of IPEK, $c_a = c_h$ and $p^{*} = 1/2$. With credibility weighting, $p^{*} > 1/2$ as soon as attackers have lower credibility than honest reporters. The update of the credibility makes this margin self-reinforcing: while the consensus is honest, attackers deviate from it, their credibility $c_a$ decreases and $p^{*}$ increases. Once $p > p^{*}$, the mechanism reverses and the honest reporters become the deviant ones. Since attackers are not spread evenly over the targets, local attacker majorities occur even when the global share is below one half; this is where the weighted consensus matters (Section~\ref{sec:results}).

\section{Evaluation methodology}\label{sec:method}

The evaluation addresses three questions: how CREST performs as a complete system against complete competing systems, how much of its benefit is due to the decision rule itself, and how much is due to the updates of the trust substrate. The protocol, including all hypotheses, arms, cells and statistical procedures, was registered before the corresponding simulation runs; deviations are reported in Section~\ref{sec:method:stats}.

\subsection{Simulation setup}\label{sec:method:setup}

The simulations are carried out with OMNeT++ 5.6.2~\cite{omnetpp-sw,omnetpp}, the Veins 5.2 framework~\cite{veins-sw,veins} and the SUMO 1.8.0 traffic simulator~\cite{sumo-sw,sumo}. The road network is a Manhattan grid. Two traffic densities are considered, a base scenario with $N_V$ vehicles and a dense scenario with $N_V^{\mathrm{d}}$ vehicles; each run lasts $T_{\mathrm{sim}}$ and uses one CA. Event, perception and attack models are those of Section~\ref{sec:system}. The ideal channel model is identical for all systems, so that differences between systems are due to their trust and decision logic. Table~\ref{tab:params} lists the parameters of the system, the attacks, the trust substrate and CREST with their values; Table~\ref{tab:params-baselines} lists those of the compared systems.

\begin{table}[p]
\centering
\caption{Parameters of the system, the attacks, the trust substrate and CREST. Values marked $^{*}$ were tuned under the fairness protocol (Section~\ref{sec:method:fair}); the published value is given in parentheses. Sensitivity values are listed after the main value. The parameters of the compared systems are given in Table~\ref{tab:params-baselines}.}
\label{tab:params}
\scriptsize
\setlength{\tabcolsep}{4pt}
\begin{tabular}{@{}L{1.9cm}L{5.9cm}L{4.5cm}@{}}
\toprule
Symbol & Description & Value \\
\midrule
\multicolumn{3}{@{}l}{\emph{Simulation}} \\
-- & Road network & $5 \times 5$ grid, $2000 \times 2000$\,m \\
$N_V$, $N_V^{\mathrm{d}}$ & Number of vehicles, base / dense scenario & 150 / 300 \\
$T_{\mathrm{sim}}$ & Simulated time per run & 1000\,s \\
$T_{\mathrm{u}}$ & Position update period of a vehicle & 1\,s \\
\midrule
\multicolumn{3}{@{}l}{\emph{Network, events and perception (Section~\ref{sec:system:net})}} \\
$r_{\mathrm{b}}$ & Broadcast range & 500\,m \\
$N_E$ & Number of events & 40 \\
$S_E$, $S_L$ & Event severity / location criticality: Class 1; Class 2; Class 3 & 0.2 / 0.2; one in $[0.7, 1.0]$, the other in $[0.1, 0.6]$; 0.9 / 0.9 \\
$D_0$ & Base active duration, $D_e = D_0(1 + CF)$ & 100\,s \\
$R$ & Impact radius of an event & 250\,m \\
$p_0$ & Base perception-error probability & 0.1 \\
$T_{\mathrm{rep}}$ & Report period of a witness & 5\,s \\
$D_{\mathrm{acc}}$ & Acceptance distance of reports & 800\,m \\
\midrule
\multicolumn{3}{@{}l}{\emph{Decision geometry (Section~\ref{sec:system:geometry})}} \\
$\beta$; $\beta_{\mathrm{L}}$ & Ring ratio of the decision ring $R < d \le \beta R$; enlarged ratio & 1.3; 2 \\
$h$; $h_{\max}$ & Relaying hops; sensitivity value & 0; 2 \\
\midrule
\multicolumn{3}{@{}l}{\emph{Attacks (Section~\ref{sec:system:threat})}} \\
$\mathcal{A}$; $\mathcal{A}_3$ & Attacker ratios $a$; subset used for Groups 2, 3, the dense scenario and the geometry analysis & $\{0.10, 0.25, 0.30, 0.35, 0.40\}$; $\{0.10, 0.25, 0.40\}$ \\
$\theta_E$, $\theta_L$ & Severity and location triggers & 0.6, 0.4 \\
$T_{\mathrm{oo}}$ & On-off period (attacks in the first half) & 40\,s \\
$r_{\mathrm{s}}$ & Radius of the spoofed-position disk around the event centre & 50\,m \\
$\delta$; $\delta^{-}$, $\delta^{+}$ & Camouflage margin; sensitivity values & 0.10; 0.05, 0.20 \\
\midrule
\multicolumn{3}{@{}l}{\emph{Trust substrate: IPEK (Section~\ref{sec:system:trust})}} \\
$GT^{0}$ & Initial global trust (vacuous mass) & 0.5 \\
$\mu$; $\omega_E$, $\omega_L$ & Balance coefficient; weights of $S_E$ and $S_L$ in the reward & 0.15; 0.6, 0.4 \\
$\kappa$ & Penalty ratio & $5^{*}$ (5) \\
$d_0$ & Deviation at zero credibility & $0.3^{*}$ (0.5) \\
$\rho$ & Credibility smoothing factor & $1.0^{*}$ (0.5) \\
$n_{\min}$ & Minimum reporters per target (and targets per reporter) & 3 \\
$\Delta$ & Round length of the global-trust update & 50\,s \\
\midrule
\multicolumn{3}{@{}l}{\emph{CREST (Section~\ref{sec:crest})}} \\
$\tau$ & Abstention threshold & 0.10 \\
$\theta$ & Detection threshold & $0.35^{*}$ \\
$W$ & Revocation window of the time rule & 100\,s \\
$N_r$ & Rounds of the round rule (IPEK; ablation) & 3 \\
$m_{\mathrm{p}}$, $\sigma_{\mathrm{p}}$, $\tau_{\mathrm{p}}$ & Plausibility check: range margin, distance-estimate noise, tolerance & 50\,m, 0.2, 1.0 \\
\bottomrule
\end{tabular}
\end{table}

\begin{table}[t]
\centering
\caption{Parameters of the compared systems. Values marked $^{*}$ were tuned under the fairness protocol (Section~\ref{sec:method:fair}); all other values are those of the original papers.}
\label{tab:params-baselines}
\scriptsize
\setlength{\tabcolsep}{4pt}
\begin{tabular}{@{}L{1.9cm}L{5.9cm}L{4.5cm}@{}}
\toprule
Symbol & Description & Value \\
\midrule
\multicolumn{3}{@{}l}{\emph{HTEMD~\cite{htemd}}} \\
$\alpha_{\mathrm{H}}$ & Penalty rate of the local update & $0.3^{*}$ (unspecified; initial guess 0.3) \\
$\lambda_{\mathrm{H}}$ & Coefficient of the activity term of direct trust & $2.0^{*}$ (unspecified; 0.5) \\
$k_{\mathrm{H}}$ & Width of the outlier filter in units of $1.4826\,$MAD & $1.0^{*}$ (unspecified; 3) \\
$GT^{0}_{\mathrm{H}}$ & Initial trust & 0.2 \\
$\tau_{\mathrm{H}}$ & Decision threshold & 0.5 (paper: drawn from $\mathcal{N}(0.5, 0.5)$) \\
$\theta_{\mathrm{H}}$ & Revocation threshold (immediate revocation) & 0.05 \\
\midrule
\multicolumn{3}{@{}l}{\emph{TCEMD~\cite{tcemd}}} \\
$\Theta_{\mathrm{T}}$ & Distinct negative reporters required for revocation & 50 \\
$\theta_{\mathrm{T}}$ & Revocation threshold on trust & 0.05 \\
$\lambda_{\mathrm{T}}$ & Trust decay without feedback & 0.95 \\
$T_p$ & Relaying threshold of a vehicle, drawn once & $\mathcal{N}(0.5, 0.5)$, clipped to $[0,1]$ \\
\midrule
\multicolumn{3}{@{}l}{\emph{AR-NDA~\cite{arnda}}} \\
$\Delta_{\mathrm{A}}$ & Belief difference for a direct decision & 0.9 \\
$n_0$; $\Delta n$ & Vote threshold; increment per importance level & 15 (sensitivity 5, 10); 10 \\
$\mathit{Thr}_1$, $\mathit{Thr}_2$ & Confidence thresholds (several reports / a single report) & 0.7, 0.6 \\
$a_{\mathrm{A}}$, $b_{\mathrm{A}}$ & Coefficients of the weighting factor $a_{\mathrm{A}} + b_{\mathrm{A}} \exp(m)$ & 0.4, 0.2 \\
\bottomrule
\end{tabular}
\end{table}

\subsection{Compared systems}\label{sec:method:systems}

The comparison is organised in three groups. Each group answers one question and changes only the components relevant to that question.

\paragraph{Group 1: complete systems}
The main comparison is between complete systems, each consisting of its own trust model, decision rule and revocation rule as specified in its original paper. Attaching the CREST decision rule to the trust models of the competitors would not test the competitors as published, and attaching their decision rules to IPEK would not test CREST; therefore no such hybrids are used.
\begin{itemize}
  \item \textbf{CREST:} IPEK trust with the substrate updates of Section~\ref{sec:crest:substrate} (credibility-weighted median consensus, time-based revocation with window $W$ and threshold $\theta$) and the abstaining decision rule with threshold $\tau$. Application results are reported for the policies CREST+T and CREST+I of Section~\ref{sec:crest:decision}.
  \item \textbf{HTEMD}~\cite{htemd}: hybrid direct and certificate trust, an entropy-weighted forced binary decision taken once when the vehicle enters the decision area, and immediate revocation when trust falls below $\theta_{\mathrm{H}}$, as in the paper.
  \item \textbf{TCEMD}~\cite{tcemd}: cascading trust, a trust-weighted forced binary decision on the sign of $D_t = \sum_i s_i\,Tr_i / \sum_i Tr_i$ with $s_i = \pm 1$, and revocation when the number of distinct negative reporters exceeds $\Theta_{\mathrm{T}}$ and trust falls below $\theta_{\mathrm{T}}$, as in the paper.
\end{itemize}

\paragraph{Group 2: decision rules on a common trust substrate}
To isolate the contribution of the decision rule, the decision rules are compared on the same trust substrate, revocation rule and threshold as CREST; only the rule that turns the received reports into a decision differs.
\begin{itemize}
  \item \textbf{Majority:} unweighted vote; the vehicle abstains when the normalised vote difference $\lvert n_{\mathrm{A}} - n_{\mathrm{N}} \rvert / n$ is below $\tau$.
  \item \textbf{Forced binary:} the CREST fusion of Section~\ref{sec:crest:decision} without abstention, i.e.\ with the threshold set to zero; a decision is taken as soon as one report is available.
  \item \textbf{Unweighted:} the CREST fusion and threshold with all weights set to one, i.e.\ without trust, proximity and freshness.
  \item \textbf{AR-NDA}~\cite{arnda}: fuzzy trustworthiness from reputation and proximity, Bayesian-approximated Dempster combination, direct decision when the belief difference exceeds $\Delta_{\mathrm{A}}$, and the paper's vote-count inference otherwise, with all parameters as given in the paper. AR-NDA defines no reputation model and assumes that reputations are available; here they are the global trust values of the common substrate. AR-NDA always reaches a firm decision when at least one report is available.
\end{itemize}

\paragraph{Group 3: ablation of the substrate updates}
CREST is compared with two variants that each undo one substrate update: CREST with the plain median consensus of IPEK ($-$weighted median) and CREST with the round-count revocation rule with $N_r$ rounds ($-$time rule).

\paragraph{Sensitivity analyses}
\begin{itemize}
  \item \textbf{Decision geometry:} the ring ratio and relaying are varied in a $2 \times 2$ design, $\{\beta, \beta_{\mathrm{L}}\} \times \{0, h_{\max}\}$ hops, for CREST, the forced binary rule, HTEMD and TCEMD.
  \item \textbf{Camouflage margin:} $\delta^{-}$ and $\delta^{+}$ in addition to the main value $\delta$, for CREST.
  \item \textbf{AR-NDA vote threshold:} two values of $n_0$ below the published one, since the number of reports per decision in the present scenario differs from that of the original evaluation.
  \item \textbf{Plausibility check:} CREST with the optional check of Section~\ref{sec:crest:plaus} under the two spoofing attacks.
\end{itemize}

\subsection{Mapping of the competing systems}\label{sec:method:mapping}

The competing systems were designed for different geometries and partly leave behaviour unspecified. Each was mapped onto the common framework with the following rules, applied identically to all systems.
\begin{itemize}
  \item \textbf{Common geometry.} All systems use the same impact radius $R$, decision ring, acceptance distance $D_{\mathrm{acc}}$ and report-age limit; the zone radii of the original papers are replaced by this geometry. The decision instant follows each paper: HTEMD and TCEMD decide once on entering the decision area with the reports collected until then, and decide on the first report if none was available at entry.
  \item \textbf{Common evidence.} Where a paper derives trust from communication success (HTEMD), which carries no information under an ideal channel, the agreement of a reported state with the judge's own perception is used instead, which is the evidence used by IPEK and TCEMD.
  \item \textbf{No roles.} All vehicles are equal. Authority classes (TCEMD) and government nodes (AR-NDA) are not modelled. All vehicles start from the same initial trust within each system: $GT^{0}$ in CREST and TCEMD, which equals the medium-authority value of TCEMD, and the prescribed value $GT^{0}_{\mathrm{H}}$ in HTEMD. This removes an advantage that both systems have in their own designs; since no other system receives role information, the comparison remains symmetric.
  \item \textbf{Time scales.} Update periods and revocation windows longer than the simulated time are mapped onto it: TCEMD updates trust every round $\Delta$ instead of with its own, longer period, and its negative-reporter count is accumulated over the whole run.
  \item \textbf{Randomised thresholds.} HTEMD draws its decision threshold at random; a fixed threshold $\tau_{\mathrm{H}}$ equal to the mean of that distribution is used so that runs remain paired across systems. The relaying threshold $T_p$ of TCEMD, drawn once per vehicle, is kept, since it matters only in the relaying sensitivity analysis.
\end{itemize}
A complete list of the mapping decisions for each system, together with the reasons, is provided in the supplementary material (Sections S1--S4).

\subsection{Fairness protocol}\label{sec:method:fair}

Parameters that a paper specifies are used as specified. Parameters that a paper leaves unspecified were tuned for every system with the same procedure: tuning seeds 10--19, the reference cell D-Sev at 25\% attackers, a criterion fixed in advance, extension of the grid whenever the optimum lies on its boundary, and a tie rule that prefers the published or default values. The selected values were frozen before any reporting seed was run. The procedure was applied as follows.
\begin{itemize}
  \item \textbf{HTEMD:} the three unspecified trust parameters $\alpha_{\mathrm{H}}$, $\lambda_{\mathrm{H}}$ and $k_{\mathrm{H}}$ were tuned for the open-loop separation of attackers and honest vehicles. With the initial guesses (Table~\ref{tab:params-baselines}) the separation on the tuning seeds was considerably lower than with the tuned values (AUC 0.627 against 0.848).
  \item \textbf{IPEK:} three parameters, matched in number and role to those of HTEMD, were tuned with the same criterion: $\kappa$, $d_0$ and $\rho$ (Section~\ref{sec:system:trust}).
  \item \textbf{CREST:} the detection threshold $\theta$ was selected for the closed-loop F1 on the tuning seeds, and the consensus variant was selected under a protection rule that forbids losses at the two lowest attacker ratios of $\mathcal{A}_3$.
  \item \textbf{TCEMD and AR-NDA:} all numerical parameters are given in the papers, so no tuning was required.
\end{itemize}
Tuning the unspecified parameters of the competitors works in their favour; with the initial guesses HTEMD performs worse than reported here. The reporting seeds, 20--39 for the base scenario and 20--29 for the open-loop, dense and geometry experiments, were not used in any earlier stage.

\paragraph{Discrepancy with the HTEMD paper}
In the HTEMD paper the trust of malicious vehicles falls to about the revocation threshold $\theta_{\mathrm{H}}$. In the present implementation, which follows the published equations, the trust of attackers does not fall that far, and with the published threshold HTEMD revokes practically no attacker while revoking honest vehicles (Section~\ref{sec:results}). The behaviour of the implementation is reported as measured, and the discrepancy is noted here rather than corrected by changing the published threshold.

\subsection{Metrics}\label{sec:method:metrics}

Metrics are defined at three levels.
\begin{itemize}
  \item \textbf{Trust level.} The \emph{open-loop} runs disable revocation and score each vehicle by $s_i = \min_k \max(GT_i(k), \dots, GT_i(k + N_r - 1))$ over its rounds; the AUC is the probability that an attacker scores lower than an honest vehicle. In the \emph{closed-loop} runs revocation is active. A true positive is a revoked vehicle that has attacked at least once, a false positive a revoked honest vehicle. The F1 score combines precision and recall over the vehicles that have attacked; the honest false-revocation rate is the share of honest vehicles revoked; the detection delay is the time from an attacker's first attack to its revocation.
  \item \textbf{Decision level.} Every passage of a vehicle through the decision ring of an event activation is an \emph{encounter}. The decision error is the share of firm false alarms and firm missed events among informed encounters, i.e.\ firm and unsure ones; uninformed encounters are excluded, since no decision rule can act without reports. The unsure rate is the share of unsure encounters among informed ones.
  \item \textbf{Application level.} A firm \emph{active} decision triggers a reroute. The \emph{missed-reroute rate} is the share of encounters with an active event that did not lead to a reroute, including unsure and uninformed encounters; the \emph{unnecessary-reroute rate} is the share of encounters with an inactive event that led to a reroute. For CREST+T and CREST+I, unsure encounters are resolved by the respective policy; for systems without abstention the policies coincide with the plain rates. The \emph{decision lead} is the distance $d - R$ of the vehicle from the impact radius at a firm active decision.
\end{itemize}
The decision error alone favours abstention, because an unsure encounter is not counted as an error. For this reason it is reported together with the unsure rate and the application-level rates, which count the cost of abstention.

\subsection{Experimental design and statistics}\label{sec:method:stats}

A \emph{cell} is a combination of attack and attacker ratio. The base-scenario closed-loop comparison of Group~1 covers the fourteen attacks at every ratio of $\mathcal{A}$ and an attack-free control, i.e.\ 71 cells with 20 seeds each. Groups~2 and~3 use the ratios of $\mathcal{A}_3$ and the control (43 cells, 20 seeds). The open-loop comparison uses the 71 cells and the dense scenario the 43 cells, both with 10 seeds. The geometry analysis uses six attacks (D-Sev, R-Sev, D-OnOff, Fab, Spoof, Camo-D) at the ratios of $\mathcal{A}_3$ and the control (19 cells, 10 seeds).

Every comparison is between CREST and one other arm on the same seeds. In each cell, the paired two-sided Wilcoxon signed-rank test is applied, and the Holm correction is applied within each family of one arm pair and one metric across the cells, at the 5\% significance level. Results are reported as the number of cells in which CREST is significantly better, not significantly different, or significantly worse, together with the cell means.

With 10 seeds the smallest attainable two-sided $p$ value of the exact test is $0.002$, and after Holm correction over 43 or more cells no cell can reach significance. This limitation of the registered protocol was identified after the runs. For the 10-seed experiments two additional tests are therefore reported and labelled as not pre-registered: a seed-level test, which averages each seed over the cells and applies one Wilcoxon test to the ten paired means, and a cell-level sign test on the cell means. The 20-seed experiments are not affected.

\section{Results}\label{sec:results}

All values are means over seeds and then over cells. Counts of the form $b/e/w$ give the number of cells in which CREST is significantly better ($b$), not significantly different ($e$) or significantly worse ($w$) than the compared arm (Section~\ref{sec:method:stats}).

\subsection{Comparison of complete systems}\label{sec:results:systems}

Table~\ref{tab:systems} summarises the comparison of CREST with HTEMD and TCEMD in the base scenario.

\begin{table}[t]
\centering
\caption{Complete systems, base scenario (closed loop: 71 cells including the control, 20 seeds; open loop: 70 attack cells, 10 seeds). Rates are fractions of encounters; Missed and Unnecessary denote the missed-reroute and unnecessary-reroute rates under the stated policy; the honest false-revocation rate is in percent.}
\label{tab:systems}
\footnotesize
\setlength{\tabcolsep}{4pt}
\begin{tabular}{@{}lccccc@{}}
\toprule
Metric & CREST & HTEMD & TCEMD & vs HTEMD & vs TCEMD \\
\midrule
\multicolumn{6}{@{}l}{\emph{Trust level}} \\
AUC, open loop$^{*}$ & \textbf{0.983} & 0.905 & 0.788 & 70/0/0 & 70/0/0 \\
F1, closed loop & \textbf{0.875} & 0.000 & 0.346 & 70/0/0 & 70/0/0 \\
Honest false revocation (\%) & \textbf{1.3} & 18.3 & 7.8 & 24/47/0 & 71/0/0 \\
Detection delay (s) & \textbf{288} & --$^{**}$ & 624 & -- & 68/2/0 \\
\midrule
\multicolumn{6}{@{}l}{\emph{Decision level}} \\
Decision error & \textbf{0.051} & 0.162 & 0.106 & 71/0/0 & 70/1/0 \\
Unsure rate & 0.057 & 0 & 0 & -- & -- \\
\midrule
\multicolumn{6}{@{}l}{\emph{Application level}} \\
Missed, no action & 0.142 & 0.210 & \textbf{0.135} & 47/7/17 & 8/18/45 \\
Missed, CREST+T & \textbf{0.100} & 0.210 & 0.135 & 69/2/0 & 55/11/5 \\
Unnecessary, CREST+T & \textbf{0.083} & 0.107 & 0.133 & 48/20/3 & 69/1/1 \\
Missed, CREST+I & \textbf{0.088} & 0.210 & 0.135 & 71/0/0 & 63/6/2 \\
Unnecessary, CREST+I & 0.126 & \textbf{0.107} & 0.133 & 5/29/37 & 28/35/8 \\
\bottomrule
\end{tabular}

\smallskip
\begin{minipage}{0.92\linewidth}\footnotesize
$^{*}$ 10 seeds; counts from the cell-level sign test, seed-level $p = 0.002$ against both (not pre-registered).
$^{**}$ HTEMD revokes practically no attacker; the delay is undefined.
\end{minipage}
\end{table}

\paragraph{Trust level}
CREST separates attackers from honest vehicles better than both competitors in every attack cell, in the open loop (AUC 0.983 against 0.905 and 0.788) and in the closed loop. With its published revocation threshold $\theta_{\mathrm{H}}$, HTEMD revokes practically no attacker (F1 at most 0.003 in every cell). Its honest false revocations are confined to the attacks with a feedback channel: they are zero in the control and under every report-only attack, but reach 18--95\% at 25\% attackers and 80--100\% at 40\% attackers under D-Sev, D-Loc, D-OnOff and D-Cont. This pattern is consistent with the trust update of HTEMD, which replaces the trust of a target by the mean of the feedback that passes the outlier filter, without history: colluding feedback that forms a local majority passes the filter and drives honest targets below the threshold. TCEMD revokes some attackers (F1 0.346) but also 7.8\% of the honest vehicles on average, including 5.5\% in the attack-free control; its global trust is recomputed in every round without history, so that a single round of negative feedback can push an honest vehicle below the threshold. CREST revokes 1.3\% of the honest vehicles on average and 0.3\% in the control; its highest false-revocation rate, 20.7\%, occurs under D-Cont at 40\% attackers, where it is still lower than that of HTEMD (100\%) and TCEMD (47.8\%).

\paragraph{Decision level}
The decision error of CREST is one third of that of HTEMD and one half of that of TCEMD (0.051 against 0.162 and 0.106) and lower in 71 and 70 of the 71 cells. The advantage grows with the attacker ratio (Fig.~\ref{fig:errratio}): under D-Cont, for example, the decision error of CREST rises from 0.043 at 25\% to 0.142 at 40\% attackers, that of TCEMD from 0.145 to 0.482 and that of HTEMD from 0.413 to 0.715. In the control, where no vehicle lies, CREST still halves the error (0.019 against 0.038 and 0.039). Here the only source of wrong reports is the perception error, which grows with the distance of the witness from the event; the proximity weight discounts exactly these reports, and the abstention threshold defers decisions that rest on them. CREST is unsure in 5.7\% of the informed encounters; HTEMD and TCEMD never abstain.

\begin{figure}[t]
\centering
\includegraphics[width=\linewidth]{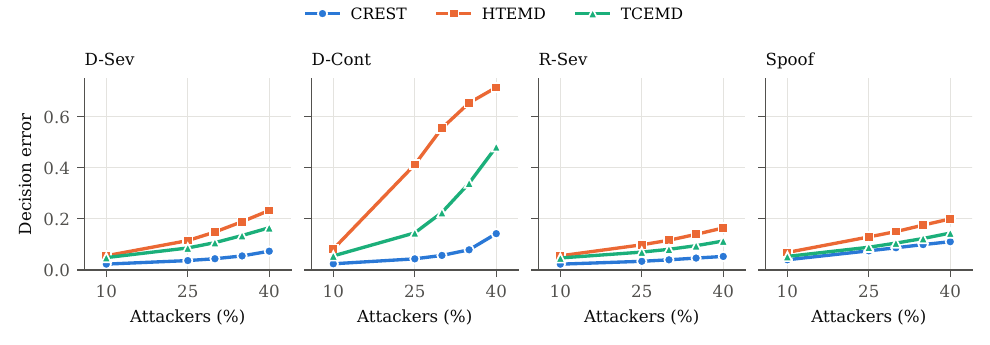}
\caption{Decision error as a function of the attacker ratio for four attacks (base scenario, 20 seeds; attacker ratios of $\mathcal{A}$). Plotted directly from the simulation results with a script prepared with the assistance of an AI tool (see the declaration of generative AI use).}
\label{fig:errratio}
\end{figure}

\paragraph{Application level}
The application-level rates show what abstention costs and what it buys. If unsure vehicles take no action, CREST produces the fewest unnecessary reroutes (0.067) but misses more hazards than TCEMD (0.142 against 0.135, worse in 45 cells): a vehicle that is still unsure at the event and keeps its route misses the event whenever it is active. Resolving the unsure encounters removes this cost. With the best-guess policy CREST+T, CREST misses fewer hazards than both competitors (0.100 against 0.210 and 0.135) and also produces fewer unnecessary reroutes (0.083 against 0.107 and 0.133); against TCEMD it is better in 55 cells and worse in 5 for missed reroutes and better in 69 cells for unnecessary ones. The cautious policy CREST+I lowers the missed-reroute rate further, to 0.088, at the price of more unnecessary reroutes (0.126); this is the expected trade-off of a policy that treats every unresolved hazard as present.

\paragraph{Attack families}
Table~\ref{tab:attacks25} breaks the comparison down by attack at 25\% attackers. Three families deserve attention.
\begin{itemize}
  \item \textbf{Position spoofing} is the only family in which CREST+T misses more hazards than TCEMD in some cells (Spoof at 25\%: 0.127 against 0.116; Spoof-OnOff at 25\%: 0.157 against 0.135). Spoofed reports claim the most reliable position and therefore receive a high proximity weight in CREST, whereas TCEMD does not use proximity at all. The decision error of CREST nevertheless remains lower than that of TCEMD in both attacks, significantly so for Spoof. The optional plausibility check addresses exactly this weakness (Section~\ref{sec:results:sens}).
  \item \textbf{Fabrication} produces the highest unnecessary-reroute rates of all attacks, for every system (CREST+T 0.148, HTEMD 0.251, TCEMD 0.253 at 25\%). Its on-off variant is also the hardest to detect for CREST (F1 0.530), because the attacker fabricates only during half of the time and only on passive events.
  \item \textbf{Trust-gated camouflage} lowers the F1 of CREST from 0.951 (D-Sev) to 0.560 (Camo-D) and from 0.964 (R-Sev) to 0.618 (Camo-R) at 25\% attackers: an attacker that knows the threshold and attacks only while its trust is high enough is revoked less often. The decision layer, however, is not deceived in the same way. The decision error under camouflage (0.032--0.036) is no higher than under the corresponding non-camouflaged attacks, since an attacker that keeps its trust above the threshold has to behave honestly most of the time.
\end{itemize}

\begin{table}[t]
\centering
\caption{Per-attack results at 25\% attackers (base scenario, 20 seeds). HTEMD has F1 $\le 0.003$ in every attack and is omitted from the F1 columns. FPR: honest false-revocation rate in percent.}
\label{tab:attacks25}
\scriptsize
\setlength{\tabcolsep}{3.5pt}
\begin{tabular}{@{}lcc@{\hspace{1em}}ccc@{\hspace{1em}}ccc@{}}
\toprule
& \multicolumn{2}{c}{F1} & \multicolumn{3}{c}{FPR (\%)} & \multicolumn{3}{c}{Decision error} \\
\cmidrule(r){2-3}\cmidrule(lr){4-6}\cmidrule(l){7-9}
Attack & CREST & TCEMD & CREST & HTEMD & TCEMD & CREST & HTEMD & TCEMD \\
\midrule
D-Sev       & 0.951 & 0.163 & 0.5 & 23.3 & 9.2  & 0.037 & 0.115 & 0.086 \\
D-Loc       & 0.979 & 0.177 & 0.6 & 33.9 & 10.4 & 0.035 & 0.141 & 0.099 \\
D-OnOff     & 0.964 & 0.182 & 0.9 & 18.0 & 9.3  & 0.037 & 0.115 & 0.087 \\
D-Cont      & 0.979 & 0.245 & 1.6 & 95.2 & 13.5 & 0.043 & 0.413 & 0.145 \\
R-Sev       & 0.964 & 0.376 & 0.3 & 0.0  & 5.4  & 0.034 & 0.098 & 0.070 \\
R-Loc       & 0.986 & 0.492 & 0.3 & 0.0  & 5.1  & 0.033 & 0.115 & 0.075 \\
R-OnOff     & 0.975 & 0.380 & 0.3 & 0.0  & 5.5  & 0.034 & 0.102 & 0.074 \\
R-Cont      & 0.996 & 0.808 & 0.3 & 0.0  & 3.3  & 0.040 & 0.173 & 0.086 \\
Fab         & 0.912 & 0.454 & 0.3 & 0.0  & 4.6  & 0.032 & 0.068 & 0.065 \\
Fab-OnOff   & 0.530 & 0.237 & 0.3 & 0.0  & 5.6  & 0.027 & 0.052 & 0.053 \\
Spoof       & 0.990 & 0.425 & 0.2 & 0.0  & 5.3  & 0.075 & 0.129 & 0.089 \\
Spoof-OnOff & 0.996 & 0.538 & 0.3 & 0.0  & 4.8  & 0.097 & 0.161 & 0.104 \\
Camo-R      & 0.618 & 0.353 & 0.3 & 0.0  & 5.4  & 0.032 & 0.098 & 0.070 \\
Camo-D      & 0.560 & 0.157 & 0.6 & 22.5 & 9.2  & 0.036 & 0.114 & 0.085 \\
\midrule
Control     & --    & --    & 0.3 & 0.0  & 5.5  & 0.019 & 0.038 & 0.039 \\
\bottomrule
\end{tabular}
\end{table}

\subsection{Decision rules on a common trust substrate}\label{sec:results:rules}

Table~\ref{tab:rules} compares the decision rules on the trust substrate of CREST. Since the trust layer, the revocation rule and the threshold are identical in all arms, the trust-level results coincide and only decision and application metrics are reported.

\begin{table}[t]
\centering
\caption{Decision rules on the CREST trust substrate (base scenario, 43 cells, 20 seeds). Missed and Unnecessary: missed-reroute and unnecessary-reroute rates under the stated policy, which affects only abstaining rules. Counts: CREST better / not different / worse.}
\label{tab:rules}
\footnotesize
\setlength{\tabcolsep}{4pt}
\begin{tabular}{@{}lccccc@{}}
\toprule
Metric & CREST & Majority & Forced & Unweighted & AR-NDA \\
\midrule
Decision error & \textbf{0.047} & 0.062 & 0.072 & 0.057 & 0.119 \\
 & & 37/2/4 & 43/0/0 & 37/0/6 & 43/0/0 \\
Unsure rate & 0.055 & 0.014 & 0 & 0.033 & 0 \\
Missed, CREST+T & 0.095 & \textbf{0.088} & 0.106 & 0.094 & 0.160 \\
 & & 4/14/25 & 43/0/0 & 12/23/8 & 43/0/0 \\
Unnecessary, CREST+T & \textbf{0.080} & 0.133 & 0.097 & 0.100 & 0.102 \\
 & & 43/0/0 & 43/0/0 & 13/30/0 & 39/1/3 \\
Missed, CREST+I & \textbf{0.084} & 0.088 & 0.106 & 0.094 & 0.160 \\
 & & 18/18/7 & 43/0/0 & 37/0/6 & 43/0/0 \\
Unnecessary, CREST+I & 0.121 & 0.136 & \textbf{0.097} & 0.100 & 0.102 \\
 & & 18/23/2 & 0/0/43 & 4/2/37 & 0/31/12 \\
Decision lead (m) & 59.2 & 59.8 & 60.2 & 58.7 & \textbf{60.8} \\
\bottomrule
\end{tabular}
\end{table}

The abstaining evidential rule has the lowest decision error of all rules. Against the forced binary rule, which uses the same fusion but never abstains, it is better in all 43 cells, both in decision error (0.047 against 0.072) and, with the CREST+T policy, in both application rates: fewer missed reroutes (0.095 against 0.106) and fewer unnecessary ones (0.080 against 0.097). Since the forced rule differs from CREST only by $\tau = 0$, this comparison isolates the effect of abstention: deferring a weak decision and acting on the evidence available later is better on both error types than acting on the evidence available at once. Pooled over all encounters, the forced rule misses 100.1 and unnecessarily reroutes 94.8 per 1000 encounters, CREST+T 96.1 and 81.0 (Fig.~\ref{fig:policies}).

The majority rule reroutes unnecessarily far more often than CREST+T (0.133 against 0.080, all 43 cells), but misses slightly fewer hazards (0.088 against 0.095). An unweighted vote counts every report equally, so near-ties are rare and the majority rule decides firmly more often; it ignores, however, which reporters are trusted and near the event. The unweighted variant of the CREST rule, which keeps the fusion and abstention but drops the weights, has a higher error (0.057) and more unnecessary reroutes (0.100) than CREST; the weights matter in addition to abstention. AR-NDA has the highest decision error (0.119) and the highest missed-reroute rate (0.160). Over all its decisions, 52\% were taken directly from the belief difference, 41\% by the final weighted comparison of the inference stage and 7\% by the single-report rule; the vote-count step was never reached. Lowering the vote threshold $n_0$ does not change this (Section~\ref{sec:results:sens}). Like the forced rule, AR-NDA always commits to a decision once a report is available and cannot defer a decision that rests on thin evidence.

The decision lead of CREST is about 1\,m shorter than that of the forced rules: a deferring vehicle decides slightly later, i.e.\ closer to the event. Relative to the width of the decision ring, $(\beta - 1)R$, this cost is small.

\begin{figure}[t]
\centering
\includegraphics[width=\linewidth]{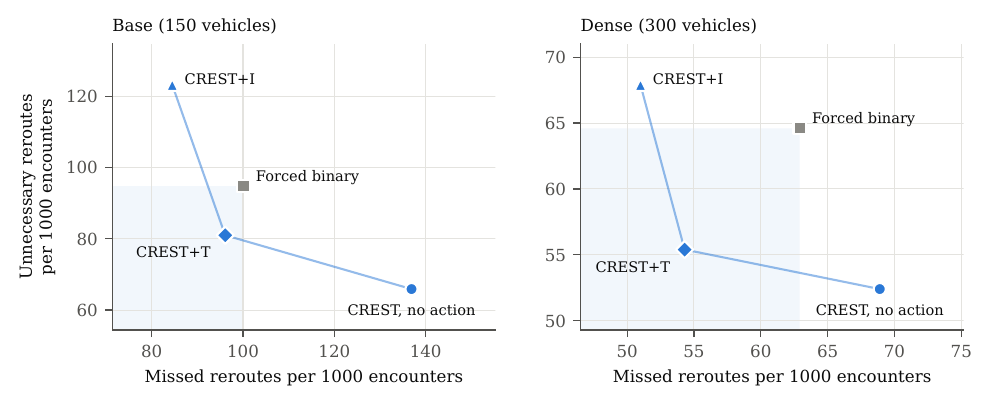}
\caption{Application outcome of the three CREST policies and of the forced binary rule, pooled over all cells and seeds. The three CREST policies share all firm decisions and differ only in the treatment of unsure encounters. The shaded area marks outcomes that are better than the forced rule on both axes. Plotted directly from the simulation results with a script prepared with the assistance of an AI tool (see the declaration of generative AI use).}
\label{fig:policies}
\end{figure}

Fig.~\ref{fig:policies} shows the trade-off between the three policies. No action minimises unnecessary reroutes (65.9 per 1000 encounters in the base scenario) and pays for it with missed hazards: 136.9 per 1000 encounters, i.e.\ 36.8 more than the forced rule. CREST+T removes this cost and is the only policy that improves on the forced rule on both axes (96.1 and 81.0 against 100.1 and 94.8). CREST+I misses the fewest hazards (84.5) at the highest rate of unnecessary reroutes (123.2) and is the choice when missed hazards are much more costly than detours. The trade-off is milder in dense traffic, where abstention is rarer (1.8\% against 5.4\% unsure encounters) and the cost of no action shrinks to 6.0 missed hazards per 1000 encounters, as predicted by Proposition~\ref{prop:abstain}.

\subsection{Ablation of the substrate updates}\label{sec:results:ablation}

Replacing the credibility-weighted median by the plain median of IPEK lowers the mean F1 from 0.874 to 0.837 and raises the honest false-revocation rate from 1.4\% to 4.4\% (43 cells, 20 seeds). The difference is concentrated in the dual-channel attacks at 40\% attackers, where local attacker majorities occur: under D-OnOff the F1 falls from 0.922 to 0.475, under D-Cont from 0.857 to 0.359. CREST is significantly better in 5 cells for F1 and 7 cells for decision error and never significantly worse. This is the breakdown predicted by Proposition~\ref{prop:wmedian}: the plain median fails once attackers form a local majority, and the credibility weighting moves the breakdown point beyond it.

Replacing the time rule by the round rule with $N_r$ rounds yields identical results in every run, as predicted by Proposition~\ref{prop:revoke} for $W = (N_r - 1)\Delta$. The time rule therefore does not change the behaviour at the operating point; its benefit is the independence of the revocation delay from the update tempo, which a fixed-tempo evaluation cannot show.

\subsection{Dense scenario}\label{sec:results:dense}

Table~\ref{tab:dense} repeats the main comparison with 300 vehicles, with all parameters and thresholds transferred unchanged from the base scenario.

\begin{table}[t]
\centering
\caption{Dense scenario (300 vehicles; 43 cells, 10 seeds). Counts: cells in which the CREST mean is better / worse (sign test); seed-level $p = 0.002$ for every count shown (not pre-registered).}
\label{tab:dense}
\footnotesize
\begin{tabular}{@{}lccccc@{}}
\toprule
Metric & CREST & HTEMD & TCEMD & Forced & AR-NDA \\
\midrule
AUC, open loop & \textbf{0.983} & 0.933 & 0.803 & -- & -- \\
 & & 42/0 & 42/0 & & \\
F1, closed loop & \textbf{0.889} & 0.000 & 0.351 & 0.889 & 0.889 \\
Honest false revocation (\%) & \textbf{0.8} & 10.5 & 14.7 & 0.8 & 0.8 \\
Decision error & \textbf{0.034} & 0.099 & 0.069 & 0.047 & 0.075 \\
 & & 43/0 & 40/3 & 43/0 & 43/0 \\
Missed, CREST+T & \textbf{0.058} & 0.123 & 0.084 & 0.067 & 0.098 \\
 & & 43/0 & 37/6 & 43/0 & 42/1 \\
Unnecessary, CREST+T & \textbf{0.057} & 0.090 & 0.104 & 0.066 & 0.076 \\
 & & 42/1 & 42/1 & 43/0 & 41/2 \\
\bottomrule
\end{tabular}
\end{table}

All conclusions of the base scenario hold in dense traffic. CREST has the lowest decision error and, with CREST+T, the lowest rates of both missed and unnecessary reroutes. Against the forced binary rule it has lower cell means in all 43 cells on both application rates (seed-level $p = 0.002$), which supports the pre-registered hypothesis that the benefit of abstention persists when evidence is abundant. As explained in Section~\ref{sec:method:stats}, the registered cell-wise test cannot reach significance with 10 seeds, and this support rests on the additional tests. More reports per decision reduce the decision error of every system, and the gap between CREST and the forced rule narrows (0.034 against 0.047), since abstention is needed less often. The honest false-revocation rate of CREST is lower than that of HTEMD on average (0.8\% against 10.5\%), but not in every cell: as in the base scenario, HTEMD revokes no honest vehicle under report-only attacks, so that its cell mean is lower in 28 of the 43 cells, while its losses under the dual-channel attacks dominate the average.

\subsection{Sensitivity analyses}\label{sec:results:sens}

\paragraph{Decision geometry}
Table~\ref{tab:geometry} varies the decision ring and relaying. Enlarging the ring ratio from $\beta$ to $\beta_{\mathrm{L}}$ lets vehicles decide about three times earlier (decision lead 168--184\,m instead of 59--60\,m) but with less evidence, which raises the decision error and the missed-reroute rate of every system. Two-hop relaying reduces missed reroutes, because fewer vehicles remain uninformed, and increases unnecessary reroutes, because false reports are relayed as well. Relaying over $h_{\max}$ hops in the enlarged ring, for example, lowers the missed-reroute rate of CREST+T from 0.120 to 0.091 and raises its unnecessary-reroute rate from 0.095 to 0.134. The ranking of the systems does not depend on the geometry. In all four configurations CREST has the lowest decision error in all 19 cells against every arm, and lower missed and unnecessary reroute rates than the forced rule in all 19 cells. Against HTEMD and TCEMD the application rates of CREST are significantly lower in 9 to 19 cells, depending on configuration and metric, and not significantly different in the others; the only cells in which CREST is significantly worse are one cell of the missed-reroute rate against TCEMD in each configuration with the main ring ratio.

\begin{table}[t]
\centering
\caption{Decision geometry (19 cells, 10 seeds). $\beta$, $\beta_{\mathrm{L}}$: main and enlarged ring ratio; $h$: relaying hops (Table~\ref{tab:params}). Application rates with the CREST+T policy for CREST.}
\label{tab:geometry}
\footnotesize
\begin{tabular}{@{}llcccc@{}}
\toprule
Geometry & Metric & CREST & Forced & HTEMD & TCEMD \\
\midrule
\multirow{4}{*}{$\beta$, $h = 0$} & Decision error & \textbf{0.044} & 0.066 & 0.114 & 0.086 \\
 & Missed reroutes & \textbf{0.086} & 0.096 & 0.146 & 0.107 \\
 & Unnecessary reroutes & \textbf{0.084} & 0.099 & 0.125 & 0.143 \\
 & Decision lead (m) & 59 & 60 & 60 & 60 \\
\midrule
\multirow{4}{*}{$\beta_{\mathrm{L}}$, $h = 0$} & Decision error & \textbf{0.065} & 0.096 & 0.158 & 0.124 \\
 & Missed reroutes & \textbf{0.120} & 0.140 & 0.207 & 0.164 \\
 & Unnecessary reroutes & \textbf{0.095} & 0.121 & 0.122 & 0.141 \\
 & Decision lead (m) & 168 & 176 & 176 & 177 \\
\midrule
\multirow{4}{*}{$\beta_{\mathrm{L}}$, $h = h_{\max}$} & Decision error & \textbf{0.066} & 0.094 & 0.153 & 0.121 \\
 & Missed reroutes & \textbf{0.091} & 0.108 & 0.173 & 0.130 \\
 & Unnecessary reroutes & \textbf{0.134} & 0.170 & 0.186 & 0.216 \\
 & Decision lead (m) & 180 & 184 & 184 & 184 \\
\midrule
\multirow{4}{*}{$\beta$, $h = h_{\max}$} & Decision error & \textbf{0.045} & 0.065 & 0.118 & 0.095 \\
 & Missed reroutes & \textbf{0.077} & 0.085 & 0.131 & 0.095 \\
 & Unnecessary reroutes & \textbf{0.105} & 0.124 & 0.200 & 0.229 \\
 & Decision lead (m) & 59 & 60 & 60 & 60 \\
\bottomrule
\end{tabular}
\end{table}

\paragraph{Camouflage margin}
A larger margin makes the camouflaged attacker more cautious. With the smaller margin $\delta^{-}$ the F1 of CREST is 0.705 (Camo-R) and 0.644 (Camo-D) at 25\% attackers; with the larger margin $\delta^{+}$ it falls to 0.427 and 0.365. The decision error stays between 0.029 and 0.036 for all margins: the more cautious the attacker, the more it behaves honestly and the less it harms the decision.

\paragraph{AR-NDA vote threshold}
With the two lowered values of $n_0$ (Table~\ref{tab:params-baselines}) the decision error of AR-NDA under D-Sev, R-Sev and Spoof at 25\% attackers is identical to that with the published value (0.110, 0.107 and 0.129). The weak performance of AR-NDA in the present scenario is therefore not caused by the choice of $n_0$.

\paragraph{Plausibility check}
Under the two spoofing attacks, the optional plausibility check lowers the decision error of CREST from 0.087 to 0.082, with a significant reduction in all six cells, and leaves the F1 practically unchanged (0.989 against 0.986).

\section{Discussion and limitations}\label{sec:discussion}

\subsection{What abstention buys and what it costs}\label{sec:discussion:abstention}

The comparison with the forced binary rule, which uses the same fusion and differs only by the abstention threshold, isolates the effect of abstention. Its benefit does not come from refusing to act at the end: a vehicle that stays unsure and takes no action misses more hazards than a forced rule (Fig.~\ref{fig:policies}). The benefit comes from deferring. A vehicle that waits for more evidence takes a firm decision later and on more reports, and this decision is more often correct; once the remaining unsure encounters are resolved by the best-guess policy, CREST+T is better than the forced rule on both error types in every cell. Abstention should therefore be understood as a waiting state with a resolution policy, not as a third answer.

Two consequences follow for the interpretation of the metrics. The decision error counts only firm false alarms and firm misses, so it favours any rule that abstains; it has to be read together with the unsure rate and the application-level rates, which count unsure encounters as missed hazards or as detours. The application policies make the remaining trade-off explicit: no action minimises detours, the cautious policy minimises missed hazards, and the best-guess policy dominates the forced rule. The choice among them depends on the relative cost of a missed hazard and a detour in the application, which the framework leaves to the application instead of fixing it implicitly. Finally, the cost of abstention depends on the evidence available. In dense traffic vehicles hold more reports and abstain less often, so that the cost of the no-action policy shrinks considerably (Section~\ref{sec:results:rules}), as predicted by Proposition~\ref{prop:abstain}.

\subsection{Remaining attack surface}\label{sec:discussion:attacks}

\paragraph{Trust-gated camouflage}
An attacker that knows the detection threshold and attacks only while its trust remains above it is revoked much less often, and the effect grows with the margin it keeps (Section~\ref{sec:results:sens}). This is a limit of any threshold-based revocation and is not specific to CREST. The decision layer is affected less: to stay above the threshold, the attacker has to behave honestly most of the time, and the decision error under camouflage is no higher than under the corresponding non-camouflaged attacks. Detection and decision quality can therefore diverge, which is a further reason to evaluate both.

\paragraph{Position spoofing}
Spoofed reports exploit the proximity weight of the decision rule. Position spoofing is the only attack family in which CREST+T misses more hazards than TCEMD in some cells, since TCEMD ignores proximity. The optional plausibility check reduces the decision error under spoofing in every cell, but the reduction is small, and forwarded reports cannot be checked because the receiver has no signal from the original sender. A reliable defence requires position verification that is independent of the reporter, for example by neighbouring vehicles or by the infrastructure~\cite{vanderheijden2019survey}.

\paragraph{Local attacker majorities}
The credibility-weighted consensus moves the breakdown point of the credibility test beyond one half of the reporters (Proposition~\ref{prop:wmedian}) but cannot remove it. The highest honest false-revocation rate of CREST occurs under continuous dual-channel attacks at the highest attacker ratio, where attackers dominate the feedback about many targets; it remains well below that of the competing systems, but it marks the regime in which the substrate, and with it the decision layer, loses its protection.

\subsection{Negative results}\label{sec:discussion:negative}

Several mechanisms were designed, registered and tested in the development stage, on seeds disjoint from the reporting seeds of Section~\ref{sec:method:fair}, and were not adopted. They are reported because they delimit what the framework can achieve.
\begin{itemize}
  \item \textbf{Adaptive update tempo.} Shortening the round length when the threat level rises lowered the mean F1 from 0.939 to 0.905, raised the honest false-revocation rate from 4.2\% to 6.8\% and multiplied the number of rounds by 2.5. Faster rounds add noise to the global trust, and slower rounds delay its convergence; Proposition~\ref{prop:revoke} makes the revocation rule tempo-independent but not the trust update itself. A variant that only slows down when the threat level is low did not pass its registered protection test at any tempo bound.
  \item \textbf{Vehicle-side conflict signal.} The conflict between the active and the inactive evidence in the decisions of a round separated rounds under attack from rounds without attack (AUC 0.851), but less well than the threat score already available at the CA (0.907), and combining the two did not improve on the threat score. The signal was therefore dropped.
  \item \textbf{Distribution-aware threshold.} A detection threshold computed in every round as a robust lower quantile of the global trust distribution collapsed at high attacker ratios: the attackers shift the median and inflate the dispersion of the distribution they are measured against, so that the threshold fell below the attackers' trust (mean threshold 0.115 and F1 0.077 in the worst cell, against 0.716 for the fixed threshold). The calibration is self-referential; its failure follows analytically once the attacker share exceeds a value that decreases with the spread of the honest trust values.
  \item \textbf{Penalising plausibility check.} A plausibility check that also penalised the sender of every rejected report reduced the decision error under spoofing, but noise-induced rejections of honest reports lowered honest trust and caused false revocations; it failed its registered protection test. The reject-only variant of Section~\ref{sec:crest:plaus} was designed after this result and is labelled as such.
  \item \textbf{Trust-weighted consensus.} Weighting the consensus by credibility multiplied by global trust, instead of credibility alone, performed worse and lost up to 0.145 AUC under continuous dual-channel attacks at the highest attacker ratio. At high attacker shares the global trust is itself inverted, attackers being more trusted than honest vehicles, while the credibility still separates them; the product carries the inverted trust into the consensus.
\end{itemize}

\subsection{Limitations}\label{sec:discussion:limits}

The evaluation has the following limitations.
\begin{itemize}
  \item \textbf{Channel and infrastructure.} The channel is ideal, and one CA with full coverage distributes the trust values. Packet loss, delay and partial RSU coverage would reduce the evidence available to a decider and would increase abstention; the size of this effect is not quantified here.
  \item \textbf{Scenarios.} One urban grid is evaluated at two densities with one event model. Highway scenarios, other event distributions and real traces remain to be tested.
  \item \textbf{Trust substrate.} CREST requires only a scalar trust per reporter, but it was evaluated on IPEK only. Its behaviour on other trust substrates has not been measured, and the analysis of Section~\ref{sec:analysis} suggests that the robustness of the decision depends on how well the substrate keeps the attackers' trust low.
  \item \textbf{Competing systems.} HTEMD and TCEMD were mapped onto a common geometry and evidence model, and their unspecified parameters were tuned; all mapping decisions are documented in the supplementary material. The present HTEMD implementation does not reproduce the low attacker trust reported in its paper (Section~\ref{sec:method:fair}).
  \item \textbf{Statistics.} For the 10-seed experiments the registered cell-wise test cannot reach significance; the reported significance relies on additional tests that were not pre-registered (Section~\ref{sec:method:stats}).
  \item \textbf{Fixed parameters.} The abstention threshold is the same for all events. A threshold that grows with the criticality of the location is a natural extension but was not evaluated.
  \item \textbf{Threat model.} Sybil identities, external attackers and attacks on the infrastructure are out of scope, and the attackers' positions can only be checked for directly received reports.
\end{itemize}

\section{Conclusion}\label{sec:conclusion}

Trust management in vehicular networks has produced many ways of scoring vehicles and messages, while the decision that a driver actually needs, whether a reported hazard is there, has received far less attention, and the schemes that make this decision force it into two outcomes. This paper has presented CREST, a vehicle-side event-decision framework that weights the reports held by a vehicle by the trust, proximity and freshness of their senders, fuses them with Yager's rule, and decides only when the fused evidence is strong enough. Abstention is treated as a waiting state, and vehicles that remain unsure are resolved by explicit application policies. The analysis shows that a single report changes the decision margin by at most its weight, which ties the robustness of the decision to the trust that the substrate assigns to the attackers, and two updates of the IPEK substrate, a credibility-weighted median consensus and a time-based revocation rule, follow from this link.

Compared with the complete HTEMD and TCEMD systems under fourteen attacks, CREST has the lowest decision error (0.051 against 0.162 and 0.106), significantly lower than both systems in 70 of the 71 cells of the base scenario. With the best-guess policy, it also misses fewer hazards (0.100 against 0.210 and 0.135) and causes fewer unnecessary reroutes (0.083 against 0.107 and 0.133). On a common trust substrate, the abstaining rule is better than the forced binary rule on both error types in every cell, which shows that the benefit comes from deferring weak decisions rather than from refusing to act. The ranking of the systems holds in dense traffic and for all four decision geometries. The weaknesses that remain are explicit: trust-gated camouflage reduces detection, and position spoofing exploits the proximity weight, which the optional plausibility check only partly offsets.

Future work includes the evaluation of CREST on other trust substrates and with realistic channel models, a criticality-dependent abstention threshold, cost-based selection of the application policy, and position verification that does not rely on the reporters, so that the proximity weight can also be used safely for forwarded reports.

\section*{CRediT authorship contribution statement}
\textbf{Ipek Abas{\i}kele\c{s}-Turgut:} Conceptualization, Methodology, Software, Validation, Formal analysis, Investigation, Data curation, Writing -- original draft, Writing -- review and editing, Visualization, Project administration, Funding acquisition.

\section*{Declaration of competing interest}
The author declares no known competing financial interests or personal relationships that could have appeared to influence the work reported in this paper.

\section*{Funding}
This work was supported by the Scientific and Technological Research Council of T\"urkiye (T\"UB\.{I}TAK) [grant number 124E017]. The funding source had no involvement in the study design, in the collection, analysis and interpretation of data, in the writing of the report, or in the decision to submit the article for publication.

\section*{Data availability}
The simulation code, the configuration files, the analysis scripts and the processed results supporting the findings of this study are available from the corresponding author upon reasonable request. The raw simulation output is not deposited in a public repository because of its volume.

\section*{Declaration of generative AI and AI-assisted technologies in the manuscript preparation process}
During the preparation of this work the author used Claude (Anthropic) to draft and edit the manuscript text, to write the analysis and plotting scripts and the code of the schematic figures, and to check the consistency between the text, the code and the results. After using this tool, the author reviewed and edited the content as needed and takes full responsibility for the content of the published article.

\bibliographystyle{elsarticle-num}
\bibliography{refs}

\end{document}